\documentclass[draftcls,letterpaper,11pt,onecolumn]{IEEEtran}
\usepackage{amsmath}
\usepackage{amsfonts,amssymb}
\usepackage{epsfig}
\usepackage{cite}
\usepackage{graphics}
\usepackage{verbatim} 
\usepackage{setspace}
\usepackage[bvect]{commands}
\usepackage{matrix}
\def\BibTeX{{\rmfamily B\kern-.05em{\scshape i\kern-.025em
b}\kern-.08em \TeX}}

\newtheorem{theorem}{Theorem}

\newtheorem{proposition}{Proposition}

\newcommand{\bh}{\mathbf{h}}

\newcommand{\eo}{\epsilon_0}
\newcommand{\ei}{\epsilon_1}
\newcommand{\doi}{\delta_{01}}
\newcommand{\dio}{\delta_{10}}

\newcommand{\intd}{\text{d}}
\newcommand{\Cg}{C}
\begin{document}
\title{
%%   The Capacity of Multicarrier Block Fading Channels with Limited
%%   Channel State Feedback
  Limited-Rate Channel State Feedback for Multicarrier Block Fading Channels
% Multi-Carrier Transmission with Limited Feedback: Optimal Power Loading 
% for Two State Sub-Channels
  \thanks{This work was supported by the NSF under grant CCF-0644344,
    DARPA under grant W911NF-07-1-0028 and
    U.S. Army Research Office under grant W911NF-06-1-0339.
    The material in this paper was
    presented in part at the IEEE International Conference on
    Communications (ICC), Beijing, China, May, 2008.}}

\author{\authorblockN{Manish Agarwal, Dongning Guo, and Michael L.
    Honig } \\
  \authorblockA{\normalsize Department of Electrical Engineering and Computer Science,
    Northwestern University \\
   % 2145 Sheridan Road, Evanston, IL 60208 USA \\
    % \{m-agarwal, dguo, mh\}@northwestern.edu
    \today
  }   }

\maketitle

%{\bf
%Suggested notation changes:
%$\sigma_z^2 \rightarrow 1$
%$\sigma_h^2 \rightarrow 1$
%$t_0 \rightarrow t$, $t \rightarrow \tau$
%$N_G \rightarrow m$
%$\eta \rightarrow q$ (in grouping section)
%$S \rightarrow P$ (subject to coefficient)
%$B \rightarrow R_f$ (subject to coefficient, do not use $B$)
%}
\vspace{-0.7in}
\begin{abstract}
  The capacity of a fading channel can be substantially increased
  by feeding back channel state information from the receiver to
  the transmitter.
% made available to
%  the transmitter through feedback can improve the capacity
%  dramatically.  
  With limited-rate feedback what state
  information to feed back and how to encode it are important
  open questions.
  This paper studies power loading in a multicarrier system using
  no more than one bit of feedback per sub-channel.
  The sub-channels can be correlated and full channel state
  information is assumed at the receiver.
  First, a simple model with $N$ parallel two-state (good/bad) 
  memoryless sub-channels is considered, where the channel state
  feedback is used to select a fixed number of sub-channels to
  activate.  The optimal feedback scheme is the solution to a vector
  quantization problem, and the associated performance for large
  $N$ is characterized by a rate distortion function. As $N$ increases,
  we show that the loss in forward rate from the asymptotic (rate-distortion)
  value decreases as $(\log N)/N$ and $\sqrt{(\log N)/N}$\,
  with optimal variable- and fixed-rate feedback codes, respectively.
%   finite number of sub-channels are also presented.
  %%   The model can incorporate correlation among sub-channels easily.
  %%   Feedback of channel state information (CSI) enables a multi-carrier
  %%   transmitter to optimize the power allocation across sub-channels.
  We subsequently extend these results to parallel Rayleigh block fading
  sub-channels, where the feedback designates a set of sub-channels, which
  are activated with equal power. Rate-distortion feedback codes are proposed
  for designating subsets of (good) sub-channels with 
  Signal-to-Noise Ratios (SNRs) that exceed a threshold. 
  The associated performance
  %%  which then share the total power evenly.
%  To maximize the rate of the forward channel,
%  Rate-distortion feedback codes are proposed 
%  to designate a subset of sub-channels, which
%  exceed a threshold signal-to-noise ratio.
%  With an optimized 
%  threshold and sufficient feedback to activate all sub-channels,
%  which exceed the threshold, the
%  resulting capacity is known to be
%  of the same order in the number of sub-channels as that achieved by
%  water-filling with full channel state information at the
%  transmitter.
  %% 
  is compared with that of a simpler lossless source coding scheme,
  which designates groups of good sub-channels, where both
  the group size and threshold are optimized.
%one adjusting the threshold
%  and the other adjusting the threshold and is also allowed to group sub-channel.
  The rate-distortion codes can provide a significant increase in forward rate
  at low SNRs.
%Both schemes perform less favorably than using the rate-distortion code
%  in terms of the total forward rate.
\end{abstract}


\section{Introduction}
\label{intro}

% we assume that the channel is known at the rx and ignore the training time.

% need for csi
% waterfilling prohobitive feedback
% yakun's ach results--still a lot
% we build upon jy's
% mention heath's, cimini's paper
% things can be adjusted to match the amount of feedback. cap as fn of csi charaterized
% max csi and max rate

% specifically accnt for F in tdd half duplex system
% results in both regimn--F and cap

% say somethin abt sensitiviy of F*

%lit survey--rajiv, djlove--describe that they design power loading vectors offline which can be a computationally expensive proceedure specially  for large amount of feedback--our scheme is much simpler, roy, others--compare an contrast--our contribution
%lit survey-cluster ofdm?

Multicarrier transmission techniques, including orthogonal
frequency-division multiplexing (OFDM), provide a convenient way to
exploit frequency diversity in multipath fading channels.
Given the total transmit power, a substantial increase in the channel
capacity can be achieved if the power allocation across the
sub-channels is adapted to channel variations \cite{TseVis05}.
% feedback of channel state information (CSI).
For example, consider the sum capacity of $N$ independent block Rayleigh fading
sub-channels with given total power or Signal-to-Noise Ratio (SNR).
If the power is equally spread over all $N$ sub-channels,
the capacity is upper bounded by the total SNR regardless of
$N$, whereas if the power is allocated according to (optimal)
water-filling, the capacity increases as $O(\log N)$ as 
$N$ increases \cite{SunHon08IT, SanNos07IT}.

\begin{comment}
For example, the sum capacity of $N$ independent block Rayleigh fading
sub-channels is upper bounded by the total signal-to-noise ratio (SNR)
if the power is equally spread over the sub-channels, regardless of
$N$.  On the other hand, the capacity achieved by optimal
water-filling increases as $O(\log N)$ as the number of sub-channels
($N$) increases \cite{SunHon08IT, SanNos07IT}.
\end{comment}

The state or quality of the sub-channels is typically measured at the
receiver and sent to the transmitter through a feedback channel.
We refer to this as {\em channel state feedback (CSF)}.
Obviously, optimal power allocation requires a prohibitive (infinite)
amount of CSF in case of continuous channel state.
Even if the channel state can be discretized, the number of
sub-channels may exceed the total number of feedback bits.
Hence, what state information to feed back and how to encode the
feedback are important questions.

This work studies the use of limited CSF for
maximizing the achievable rate of multicarrier block fading channels.
It is assumed that the sub-channel states are known or can be
measured accurately at the receiver.  The channel state is encoded
using fewer than one bit per sub-channel and then sent to the
transmitter through a noiseless feedback channel.  The transmitter
chooses a subset of sub-channels to activate
%allocates the power to the sub-channels 
based on the feedback.

The problem of encoding the feedback is essentially a vector
quantization (VQ) problem, where the channel state is
mapped to a given number of bits for later reconstruction.
Unlike the usual quantization problem, the reconstruction here is to
produce a power loading vector for the sub-channels, where the distortion
metric is the gap between the rate achieved using the feedback
and the capacity achieved with known channel state at the transmitter.

Multicarrier power allocation with limited-rate feedback has been
previously considered in \cite{LovHea05TVT, RonVor06TC, SunHon03GC,
  SunHon08IT}.  In particular,~\cite{LovHea05TVT} applies the Lloyd
algorithm to produce a codebook of power loading vectors, which
maximizes an objective
such as achievable rate.  Unfortunately,
the size of the codebook in \cite{LovHea05TVT}, and hence the search
complexity, grows exponentially with the amount of feedback. 
Other heuristic schemes with one bit feedback per sub-channel 
have been proposed in \cite{SunHon08IT, RonVor06TC, SanNos07IT}.

This paper investigates the %fundamental
trade-off between the forward data rate and the amount of CSF
for block-fading multicarrier channels assuming \emph{no more
than one bit} of feedback per sub-channel.
Furthermore, in contrast with the lossless feedback source coding schemes
analyzed in \cite{SunHon08IT}, here we consider the more general class of
lossy (rate-distortion) source codes.

% Note that CSF is quite expensive for a large number of sub-channels.
\begin{comment}
it is shown in \cite{SunHon08IT} that the order-optimal order $O(\log
N)$ can in fact be achieved using $O(\log^3N)$ bits of feedback per
coherence block.
It is shown in \cite{SunHon08IT} that with this type of feedback scheme
and $N$ independent block Rayleigh fading sub-channels,
$O(\log^3N)$ bits of feedback per coherence block can
achieve the optimal growth rate of $O(\log N)$ bits/channel use.
\end{comment}

We first consider a model with two fading states only.
% Two block-fading channel models are considered.  The first model assumes that
Each sub-channel randomly assumes either a {\em good} or {\em bad} state 
during a coherence block.  
% The state is a randomly drawn Bernoulli random variable.
For the case of independent two-state sub-channels studied in
Section \ref{s:2}, the role of the feedback is to direct the
transmitter to select as many good sub-channels as possible to
activate subject to the power constraint.
The fundamental trade-off between the feedback rate and the sum
capacity can be characterized using rate distortion theory in the
limit of infinite number of sub-channels.  For given finite number
of sub-channels, we also quantify the gap between rates achievable by
random coding and the rate distortion bound.
%% The rate distortion function is also an upper bound for any finite
%% number of sub-channels, whereas a corresponding lower bound is
%% developed in this work.
Specifically, with variable-rate feedback codes the gap
decreases as $(\log N)/N$, whereas with fixed-rate codes the gap
decreases as $\sqrt{(\log N)/N}$.

We also compare the rate-distortion approach
with a simple lossless source coding scheme, which
reports as many good sub-channels as the feedback rate allows.
Numerical plots show that good codes in the rate distortion sense
typically achieve much higher forward rate.
The result is then extended to the case of
correlated two-state sub-channels in Section \ref{s:2corr}, where 
the sub-channel states are assumed to form a Markov chain.  
Upper and lower bounds on the
forward rate are derived as a function of the feedback rate. 

\begin{comment}
!!!!!!!!!!!!!!!!!!!!!!!!!!!!!!!!!!!!!!!!!!!!!!!
We note that previous work which uses rate distortion theory to 
study the trade-off of capacity and CSF rate in the context 
of multiple input multiple output (MIMO) 
channels have been presented in \cite{DabLov06IT,ZhaGuo07TVT}.
\end{comment}

With the insights gained from the two-state channel model, we then
study the problem of limited CSF for Rayleigh fading
sub-channels.  The fading coefficient, or state
of each sub-channel is a
Circularly Symmetric Complex Gaussian (CSCG) random
variable during each coherence block.
The case of independent sub-channels is studied in Section \ref{s:ray}
whereas the case of correlated sub-channels is discussed in Section
\ref{s:raycorr}.
The state of each sub-channel is first reduced to
a binary variable by comparing its gain with a threshold.
Similar feedback codes as considered for the two-state channels is
used to instruct the transmitter which sub-channels to activate,
assuming the power is distributed evenly over the activated sub-channels.
The threshold is selected to maximize the forward rate 
given a fixed feedback rate.
It turns out that the trade-off
admits a similar characterization as that for two-state sub-channels.
Although reduction of Rayleigh states to binary states
induces loss, the scheme with optimized threshold
and a moderate amount of feedback performs close to optimal
water-filling with channel coefficients known at the transmitter.
%sufficient feedback to activate all the sub-channels that exceed 
%the threshold achieves
In particular, given a total power constraint,
the scheme can achieve a forward rate, which has the
same order of increase with the number of sub-channels
as that of water-filling~\cite{SunHon08IT}.

Two heuristic lossless source schemes for the reduced (two-state)
version of the Rayleigh channel are also considered 
for comparison in Section \ref{s:ray}.
In particular, in one of the schemes, taken from
\cite{CheBer07ISIT}, the sub-channels are divided evenly into
groups and the feedback indicates the set of groups in which all
sub-channel gains exceed the threshold.
A binary state vector, indicating which groups to activate,
is then compressed using lossless source coding and fed back
to the transmitter.
The group size and threshold can be adjusted to maximize
the forward achievable rate, subject to the feedback rate constraint.
%
%%%% reduce the following in final revision
Such grouping, or clustering, of sub-channels to reduce feedback
overhead has also been studied in \cite{TanHea06WPMC} in a multiuser
setting.  Clustering sub-channels to reduce the training overhead and
peak-to-average power ratio was previously studied in
\cite{CimDan96GC}.
We characterize the growth in achievable rate with the number of
sub-channels (for large $N$) as a function
of the amount of feedback (which can also scale with $N$). Numerical
examples show that the analytical results 
are quite accurate for finite-size
systems of interest.  In general, these heuristic schemes achieve a smaller 
forward rate than for the rate-distortion schemes, 
given a fixed feedback rate.

%The forward/feedback rate trade-off is generally
%inferior to that of good codes in the rate distortion sense.

%% analysis shows that these results are quite accurate for a finite size
%% system of interest, namely, when $N$ is a few hundred.

\section{Independent Two-state Sub-channels}
\label{s:2}

Consider a bank of $N$ independent and statistically identical block
fading sub-channels.  During each coherence block, each sub-channel
randomly takes one of two states, namely ``good'' and ``bad,'' which
is known to the receiver.  The input is constrained such that up
to a fraction $p$ of the sub-channels can be activated by the
transmitter.  Suppose on average the amount of CSF is limited to $R_f$ bits
per sub-channel per coherence block.  The
problem is to design a feedback scheme to maximize the forward data
rate, i.e., to activate as many good sub-channels as possible.

\subsection{The Fundamental Trade-off via Rate Distortion Theory}
%\subsection{The Optimal Feedback Scheme}
\label{dmc:rd}

Let the state of sub-channel $i$ be denoted by a Bernoulli random
variable\footnote{The following convention will be adopted
  throughout the paper: A boldface letter represents a 
  vector. An uppercase letter represents a
  random vector or variable (e.g., $\S$, $S_i$), and the 
  corresponding lower case letter represents a 
  specific realization (e.g., $\s$, $s_i$).
  In addition, $\log(\cdot)$ denotes natural logarithm.} 
  $S_i$, with the
probability of being a good state denoted as $\Probk{S_i=1}=q$.
%%Let $\bS$ denote the $N \times 1$ vector denoting channel state
%% with the $i^{th}$ entry $s_i$ being "1" if the $i^{th}$ sub-channel
%% is "good" and "0" otherwise. Namely, for any $i$,  
%% $\Pr\{s_i = 1\} = q$.
Further, let %% $\hat{\bS}$ denote the input
the power loading variable $\hat{S}_i = 1$ if the $i^{th}$ sub-channel
is chosen to be activated and $\hat{S}_i = 0$ otherwise.  
Constrained by the feedback and transmission power, a feedback
scheme specifies a mapping from the set of binary channel state
vectors, whose Hamming weight is no greater than $pN$.

It is easy to see that the feedback scheme is no different than vector
quantization, where the channel state vector $\SSS=[S_1,\dots,S_N]$ is
mapped to $N R_f$ bits for recovery at the transmitter.
Constrained by the feedback rate, the
reconstruction may be prone to errors, and the quantization scheme
should be designed to achieve as few errors 
(or, as small a \emph{distortion}) in reconstruction as
possible.
%%
%%We discuss practical codes in Section \ref{s:codes}. 

The fundamental trade-off of the forward and feedback rates as
$N\rightarrow \infty$ can be addressed using rate distortion theory.
The source is a sequence of independent and identically distributed
(i.i.d.) Bernoulli$(q)$ random variables,
$S_1,S_2,\dots$.
The distortion measure can be described as $d_N(\sss,\hat{\sss}) =
\frac1N \sum^N_{i=1} d(s_i,\hat{s}_i)$ with
\begin{equation}  \label{eq:dssi}
  \textstyle d(s,\hat{s}) = \indicator{s>\hat{s}} =
  \begin{cases}
    1, \quad \text{if } s=1 \text{ and } \hat{s}=0, \\
    0, \quad \text{otherwise}.
  \end{cases}
  % \indicator{s>\hat{s}}.
\end{equation}
The metric accounts for \emph{missed opportunities}, i.e., good
sub-channels which are not activated, but does not penalize
%{\em false alarms}, i.e., 
activation of bad sub-channels, which we refer to
as {\em misfires}.
Further, the power loading vector has to satisfy a normalized weight
constraint:
\begin{equation}  \label{eq:wt}
  %\frac1N \sum^N_{i=1} \hat{s}_i \le p.
  w(\hat{\sss}) = \frac1N \sum^N_{i=1} \hat{s}_i \le p.
\end{equation}
This additional challenge of incorporating the weight constraint on the
reconstruction distinguishes the problem from the classical rate
distortion problem concerning i.i.d.\ source and single-letter
distortion measure.
Though not obvious, the rate distortion problem
admits the following simple single-letter characterization.

\begin{theorem} \label{th:rd}
  For an i.i.d.\ Bernoulli($q$) source, given the weight constraint on
  every binary reconstruction, $w(\hat{\sss}) \le p$, and the distortion
  measure, $d(s,\hat{s}) = \indicator{s>\hat{s}}$, the rate distortion
  function is
  \begin{equation}    \label{eq:RD}
    R(D) = \min_{ P_{\hat{S}|S}: \, \substack{ \Exp \, d(S,\hat{S}) \le D \\
    \Prob\{\hat{S}=1\}\le p } } I(S;\hat{S})
%%, \quad 0\le D \le (1-q)p+(1-p)q.
%%    R(p, D) = I(s_i;\hat{s}_i)
    %\,\,\, \textrm{such that} \,\, \Pr\{d(s_i,\hat{s}_i) = 1\} = \eo
  \end{equation}
  %%such that, $\Pr\{d(\hat{s}_i,s_i) = 1\} = D$ and $\Pr\{\hat{s}_i = 1\} =p$.
  where $S\sim$ Bernoulli$(q)$.
\end{theorem}

\begin{proof}
  The achievability part of the theorem is based on Shannon's
  random coding technique (see e.g., \cite{CovTho06}).  Fix
 $P_{\hat{S}|S}$ and some  $0<\delta_1<p$, which satisfy $\Exp\,
  d(S,\hat{S}) \le D$ and $\Prob\{\hat{S}=1\} \le p-\delta_1$.
  The code book of $2^{NR_f}$ codewords can be produced randomly with
  the marginal distribution $P_{\hat{S}}$.
  Further, an exponentially small fraction of codewords which violate
  the weight constraint \eqref{eq:wt} are purged.  It
  can be shown that for sufficiently large code length $N$, the random
  codebook achieves the distortion $D$ as long as the rate $R >
  I(S_i;\hat{S}_i) + \delta_2$.  The achievability part is thus proved
  because $\delta_1$ and $\delta_2$ can be chosen to be arbitrarily
  small.
  
  Showing the converse requires incorporating the weight constraint
  \eqref{eq:wt} into the standard technique of \cite{CovTho06}.
  Let $\hat{\S}$ represent the reconstruction of the random source vector $\S$.
  Consider any code of length $N$ with rate $R$ which satisfies the
  distortion and average weight constraints $\Exp[w(\hat{\S})] \le p$,
  which is a weaker than required in the theorem.\footnote{Theorem 
  \ref{th:rd} continues to hold even if the instantaneous input 
constraint \eqref{eq:wt} is replaced by an average 
constraint, namely $\Exp[w(\hat{\S})] \le p$.}    Then, due to the
  data processing theorem and the independence of $S_i$,
  \begin{align}
    NR % &\ge H(\hat{\S}) \\
    &\geq I(\S;\hat{\S}) \\
    &\geq \sum^N_{i=1} I(S_i; \hat{S}_i) \\
    \label{eq:shs}
    &\geq \min_{ P_{\hat{\SSS}|\SSS}: \, \substack{ \expect{ \frac1N
          \sum^N_{i=1} \hat{S}_i } \leq p \\
        \expect{ \frac1N \sum^N_{i=1} d(S_i,\hat{S}_i) } \leq D}}
    \sum^N_{i=1} I(S_i; \hat{S}_i).
                                %      \Prob\{\hat{X}=1\}\le p } }
    %%  &\geq \min_{P_{\hat{\S}|\S}: \, \frac1N \sum^N_{i=1}
    %%    P_{\hat{S}_i}(1) \le p} \sum^N_{i=1} I(S_i; \hat{S}_i).
  \end{align}
  The key task here is to break down the constraints on the
  distribution of the vector $\hat{\S}$ in \eqref{eq:shs} into
  constraints on the individual random variables.
  Note that $P_{\hat{\S}}$ is linear
  in $P_{\hat{\S}|\S}$ because the source distribution $P_{\S}$
  is fixed.  An important fact is that
  $I(S_i;\hat{S}_i)$ is convex in the distribution
  $P_{\hat{S}_i|S_i}$.  Because of the symmetry in the indexes $i$,
  any optimal distribution $P_{\hat{\S}|\S}$ that achieves the minimum
  of \eqref{eq:shs} must be symmetric over all indexes $i$.  Otherwise
  replacing all of them by their average yields smaller mutual
  information.
%  The weight constraint requires that
%\begin{equation}  \label{eq:wtN}
%  \frac1N \sum^N_{i=1} \Prob\{ \hat{S}_i=1 \} \le p,
%\end{equation}
%which constrain the mixture of the marginal distribution of
%$\hat{S}_i$, $i=1,\dots,N$.  Following \cite{CovTho06}, the following
%can be established:
  Therefore, due to the symmetry and the additive nature of the
  constraints, \eqref{eq:shs} implies that the rate $R$ is lower
  bounded by $R(D)$ given in \eqref{eq:RD}.
  \begin{comment}
  \begin{align}
  R\, &\geq \min_{ P_{\hat{S}|S}: \, \substack{
      \Prob\{\hat{S}=1\} \le p \\
      \Exp d(S,\hat{S}) \le D}} I(S; \hat{S}).
%  NR &\geq \sum^N_{i=1} \min_{P_{\hat{S}_i|S_i}: 
%    P_{\hat{S}_i}(1) \le p} I(S_i;\hat{S}_i) \\
%  &\geq \sum^N_{i=1} R\left( \expect{ d(S_i,\hat{S}_i) } \right) \\
%  &\geq N R\left( \frac1N \sum^N_{i=1} \expect{ d(S_i,\hat{S}_i) } \right) \\
%  &\geq N R\left( \expect{ d_N(\S,\hat{\S}) } \right) \\
%  &\geq N R(D).
  \end{align}
  \end{comment}
%  Hence the proof of the converse of the theorem.
%% However, fixing the probabilities $\Pr\{\hat{s}_i = 1\} = p$ and 
%% $\Pr\{d(s_i,\hat{s}_i) = 1\} = q \Pr\{\hat{s}_i =0 | s_i =1\} = D$
%% fixes the mutual information $I(s_i; \hat{s}_i)$.
\end{proof}

The minimization over the conditional distribution $P_{\hat{S}|S}$ in
\eqref{eq:RD} is equivalently over the crossover probabilities:
\begin{equation}    \label{eq:e01}
  \eo=P_{\hat{S}|S}(0|1)\; \text{ and }\;
  \ei=P_{\hat{S}|S}(1|0),
\end{equation}
where $\eo$ represents the probability of missing a good
sub-channel.
The mutual information $ I(S;\hat{S})$
% = H(\hat{S}) - H(\hat{S} | S)$ 
can be expressed as the following function of $(\eo,\ei)$:
%\begin{align}  \label{eq:i01}
%  I(S;\hat{S}) &= H(\hat{S}) - H(\hat{S} | S) \\  
%  &= H_2(p) - q H_2(\epsilon_0) - (1-q) H_2(\epsilon_1)  \label{eq:rrfin}
%\end{align}
\begin{equation}\label{eq:rrfin}
  i(\eo, \ei) \triangleq H_2(p) - q H_2(\epsilon_0) - (1-q) H_2(\epsilon_1),  
\end{equation}
where $H_2(\cdot)$ stands for the binary entropy function.
Unless stated otherwise, the units of all information metrics
are bits.
Note that the weight constraint \eqref{eq:wt} should be tight at the
minimum because there is no penalty on misfires.  Thus the optimal
crossover probabilities satisfy $q(1-\epsilon_0) + (1-q)\epsilon_1 =
p$. %% and $q \epsilon_0 = D$.

%% The probability of missed opportunity is then $\Exp\, d(S,\hat{S}) = q
%% \epsilon_0$.  

%% Let us denote the distortion rate function corresponding to
%% \eqref{eq:RD} as $D(R)$.
Let the capacity of a good sub-channel be $C_1$ and the capacity of a
bad sub-channel be $C_0<C_1$.  
% Then given the feedback rate $R_f$ per
% sub-channel per coherence block, 
The average number of active good sub-channels is $Nq(1-\epsilon_0)$.
%% Hence the average sum capacity is
%% \begin{align}\label{eq:cbin-fin}
%%   N C %%&= (q-q\epsilon_0) C_1 + (p-q+q\epsilon_0) C_0 \\
%%   &= N q(1-\epsilon_0) (C_1-C_0) + N p C_0.
%% %%   C(R_f) &= (q-D(R_f)) C_1 + (p-q+D(R_f)) C_0 \\
%% %%   &= (q-D(R_f)) (C_1-C_0) + p C_0.
%% \end{align}
The trade-off between the capacity and the feedback rate is
characterized by as follows.

\begin{proposition}\label{prop:RDopt}
  Given $p$, $q$, and the feedback rate $R_f$ bits per sub-channel per
  coherence block, the maximum achievable forward data rate per
  sub-channel is
  \begin{equation}    \label{eq:C}
    C = q(1-\eo^*) (C_1-C_0) + p C_0,
  \end{equation}
  where the optimal proportion of missed good sub-channels $\eo^*$ is
  the solution to the following optimization problem:
  \begin{subequations} \label{eq:opt}
  \begin{align}
    %\text{maximize:} & \quad C = q(1-\epsilon_0) (C_1-C_0) + p C_0 \\
    \text{minimize:} & \quad \eo \\ % C = q(1-\epsilon_0) (C_1-C_0) + p C_0 \\
    \text{subject to:} 
    &\quad H_2(p) - q H_2(\epsilon_0) - (1-q) H_2(\epsilon_1) \le R_f, 
    \label{eq:sp-fb} \\
    & \quad q(1-\epsilon_0) + (1-q)\epsilon_1 = p, 
    \label{eq:qep} \\
    & \quad 0 \le \epsilon_0, \epsilon_1 \le 1.
  \end{align}
  \end{subequations}
\end{proposition}

\begin{comment}
 Without loss of generality the
feedback code can be assumed to satisfy the following distortion
constraint
\begin{equation}\label{eq:defeo}
\frac1NE_{\S}\left[\sum_{i = 1}^N d(s_i,\hat{s}_i)\right] = D.
\end{equation}
where the expectation is over realizations of the state 
vector $\S$ and due to \eqref{eq:wt}, $\max \{q-p,0\} \le D \le q$.  
The capacity per sub-channel can be written as follows
\begin{align}\label{eq:cap_bscrd1}
C_{rd} &= \frac1N E_{\S} \left[C_g \sum_{i =1}^N \left(s_i  - d(s_i, \hat{s}_i)\right)
 + C_b \left(pN - \sum_{i =1}^N \left(s_i  - d(s_i, \hat{s}_i)\right)\right) \right]\\
 \label{eq:cap_bscrd2}
 C_{rd} &= C_g(q - D) + C_b (p - q + D).
\end{align}

The required feedback rate \eqref{eq:RD} can be computed as,
\begin{align}\label{eq:rriid}
R(p, q\eo) &= I(s_i ; \hat{s}_i) \\
&= H(\hat{s}_i) - H(\hat{s}_i | s_i)\\
\label{eq:rrfin}
&= H_2(p) - q H_2(\epsilon_0) - (1-q) H_2(\epsilon_1)
\end{align}
where $H(.)$ denotes the entropy of a random variable. 
Given a feedback rate of $R_f$ bits per sub-channel per channel use,
mapping probabilities $\eo$ and $\ei$ should be chosen to 
satisfy the feedback constraint,
\end{comment}

The optimization problem \eqref{eq:opt} can be easily solved numerically.
%%  Some experimental results are reported in Section \ref{s:2nr}.
Clearly, the maximum forward data rate increases as the feedback rate
increases, but the return vanishes beyond a certain point.  The minimum
feedback rate necessary for achieving the capacity can be
determined by tentatively removing the feedback constraint
\eqref{eq:sp-fb}.
If $p\ge q$, one can activate all good sub-channels so that $\eo =0$
with $\ei = \frac{p -q}{1-q}$, whereas if $p<q$, then $\eo$ can be as
small as $1-\frac{p}{q}$ by choosing $\ei=0$.
Substituting these values into \eqref{eq:rrfin}, the forward rate 
saturates at the maximum feedback rate,
\begin{equation}\label{eq:rreqmax}
  \overline{R} = \begin{cases}
      %%  H_2(p) - qH_2\left(1-\frac{p}{q}\right) & \textrm{if } p \le q,\\
  H_2(p) - qH_2\left(\frac{p}{q}\right), & \textrm{if } p \le q,\\
  H_2(1-p) - (1-q)H_2\left(\frac{1-p}{1-q}\right), & \textrm{if } p>q.
  %%  H_2(p) - (1-q)H_2\left(\frac{p-q}{1-q}\right), & \textrm{if } p>q.
  \end{cases}
\end{equation}
\begin{comment}
Interestingly, for fixed $q$, \eqref{eq:rreqmax} is a non-monotonic function 
of $p$ and maximum feedback
requirement corresponds to an intermediate input constraint.
\end{comment}
For any $R_f \le \overline{R}$, the constraint \eqref{eq:sp-fb} is tight
and $\eo^*$ can be calculated by solving the simultaneous equations
\eqref{eq:sp-fb} and \eqref{eq:qep}, which can be easily reduced to a
fixed-point equation.

\subsection{Performance Bounds for Finite Number of Sub-channels}

Proposition \ref{prop:RDopt} characterizes the asymptotic trade-off as
the number of sub-channels $N$ goes to infinity.  For a
practical situation with finite $N$, the result
needs refinement.
Note that the solution to Proposition \ref{prop:RDopt}
provides an upper bound on the forward data rate for finite $N$,
because the converse shown in the proof of
Theorem \ref{th:rd} holds for all $N$.
In the following, we consider random feedback codes and derive a
lower bound for the achievable forward data rate with given feedback
constraint.

\subsubsection{Fixed-Length Constant-Composition Feedback Code}

%% We propose the following codebook design assuming that the feedback
%% comprises of $N\,i(\eo, \ei)$ bits during each coherence block, where
%% $\eo$ and $\ei$ are obtained by solving the optimization problem in
%% Proposition \ref{prop:RDopt}.

Note that the solution to Proposition~\ref{prop:RDopt} upper bounds the
forward data rate with the average input power constraint
$E[w (\hat{\S})] = p$ and \emph{average} feedback rate
of $N i(\eo, \ei)$ bits per coherence block. Here we impose 
two additional constraints without loss of generality:
1) The binary reconstruction vectors have constant 
composition, that is, $w(\hat{\s}) = p$ for all the 
vectors $\hat{\s}$ in the feedback codebook; and 2) The feedback is
at most $N i(\eo, \ei)$ bits every coherence block. The second 
restriction implies that there are at most $2^{i(\eo, \ei)N}$
codewords.

The following proposition gives a lower bound on the
achievable forward rate given these additional constraints. 

\begin{proposition}\label{prop:bin-fix}
Let the number of feedback bits per coherence block be 
$N\,i(\eo,\ei)$, where $\eo, \ei > 0,  \not = 0.5$ are the solution
to \eqref{eq:opt}. 
Then $\exists \, N_o < \infty$ such that for $N \ge N_o$, 
the ergodic capacity achieved with the fixed-length constant-composition 
feedback code is lower bounded as
\begin{equation}\label{eq:cfix}
  C_{fixed} \ge \left( 1 - 2\sqrt{\frac{\log(Nq)}{Nq}}\right) C
\end{equation}
where $C$ is given by \eqref{eq:C}. 
\end{proposition}

The proof is given in Appendix \ref{ap:prop:bin-fix}.
The proposition implies that, for large enough $N$, 
the difference between the upper bound \eqref{eq:C} 
and achievable forward rate per sub-channel approaches 
zero at the rate $O\big(\sqrt{(\log N)/N}\,\big)$.
This further implies that the sum rate across the 
$N$ sub-channels incurs a loss, which increases as
$O\big(\sqrt{N \log N }\,\big)$
compared to $NC$. 
The proof basically follows the random coding technique of 
Goblick \cite{Goblic62} for analyzing the convergence rate 
of the rate distortion function for general sources and fixed-length block 
codes.
The contribution in this work is to incorporate the additional 
constant-composition
constraint  and to simplify the analysis by exploiting  
the binary structure of the source and the reconstruction.\footnote{We
avoid the use of complicated \emph{partition functions} in \cite{Goblic62} by 
using a Chernoff bound to evaluate the tail probability distributions.} 

%%instantaneous power constraint, fixed length
%%let \eo, \ei be soln to prop1, revise prop 2
%% percentage loss for large N, finite n given by eq 
%% proof follows goblick-diff
%% source of loss 
%% given estimate for loss, direct to plots
%% next we relax constant-p constraint and fixed block length constraint

%Along the lines of Appendix \ref{ap:prop:bin-fix},
%it is straing forward to show that, when either $\eo=0$ or,
%$\ei = 0$, we have $K_1 \approx 1/\log N$. Furthermore, the lower 
%bound \eqref{eq:cfix} is clearly loose 
%for the special case when both $\eo, \ei \to 0.5$. 

Although the result in Proposition \ref{prop:bin-fix} is stated
for large $N$, more refined lower bound is
derived in Appendix \ref{ap:prop:bin-fix} which holds
for any finite $N$. 

\subsubsection{Variable-Length % Variable Composition
 Feedback Codes}

We note that fixed-length codes can \emph{cover} a subset of most
probable channel state vectors, but are unable to adapt to deviations
from typical channel conditions.
In the following, we analyze the performance of variable-length
feedback codes.
A variable amount of feedback is allowed during each coherence block
as long as the average number of feedback bits is $N i(\eo, \ei)$.
The instantaneous power constraint is replaced by an average power
constraint.
 
\begin{comment}
In addition, we drop the constant composition constraint on the
reconstruction vectors.

Clearly, the performance suffers due to the fixed length
and constant composition constraints. In other words, a 
fixed length code will \emph{cover}
only the \emph{typical} channel state vectors, due to which
it is unable to adapt to deviations from typical channel conditions. 
For example, consider the special case
when $p = q$ and $\eo, \ei =0$. This case 
corresponds to the feedback requirement
of $N H(q)$ bits per coherence block. 
%Again,
%along the lines of Appendix  \ref{ap:prop:bin-fix},
%the loss due to the channel realizations $\S \not \in \mathbf{\mathcal{T}}_{\S;q}$
%can be shown to be $O\left(\sqrt{\frac{\log N}{N}}\right)$. 
Again, the loss compared to $C$ is of order $O\left(\sqrt{\frac{\log N}{N}}\right)$.
However, it is well known that with variable length codes 
of average length $NH_2(q)$  bits per 
coherence block, it is possible to achieve zero 
distortion and hence an ergodic capacity of 
$C$ per sub-channel. 
\end{comment}

\begin{proposition}\label{prop:bin-var}
  Let the average number of feedback bits per coherence block be
  $Ni(\eo,\ei)$, where $\eo, \ei > 0$ are the solution to Proposition
  \ref{prop:RDopt} and $\eo, \ei \neq \frac12$.  Then $\exists \, N_1
  < \infty$ such that for every $N \ge N_1$, the ergodic forward rate
  achieved with a variable-length feedback code, under an average
  input power constraint $E[w(\hat{\S})] = p$, is lower bounded as
  \begin{equation}\label{eq:cap-var}
    %C_{var} \ge \left(1-K_0\frac{\log_2 N}{N}\right) C
    C_{var} \ge \left(1-
 \frac{6}{q[H_2'(\eo)+ H_2'(\ei)]} 
\frac{\log_2 N}{N}\right) C
  \end{equation}
  %where $K_0 = \frac{6}{q[H_2'(\eo)+ H_2'(\ei)]} $ and $C$ is as given
  where $C$ is given in \eqref{eq:C}.
\end{proposition}

The proof is provided in Appendix \ref{ap:prop:bin-var}.
We follow the technique given by Pinkston \cite{Pinkst67},
albeit with slight modifications.
The main difference lies in incorporating the average input 
power constraint and avoiding the use of partition
functions by exploiting the binary structure of the state and 
power loading vectors. 

\begin{figure}
\begin{center}
\includegraphics[width = 5in]{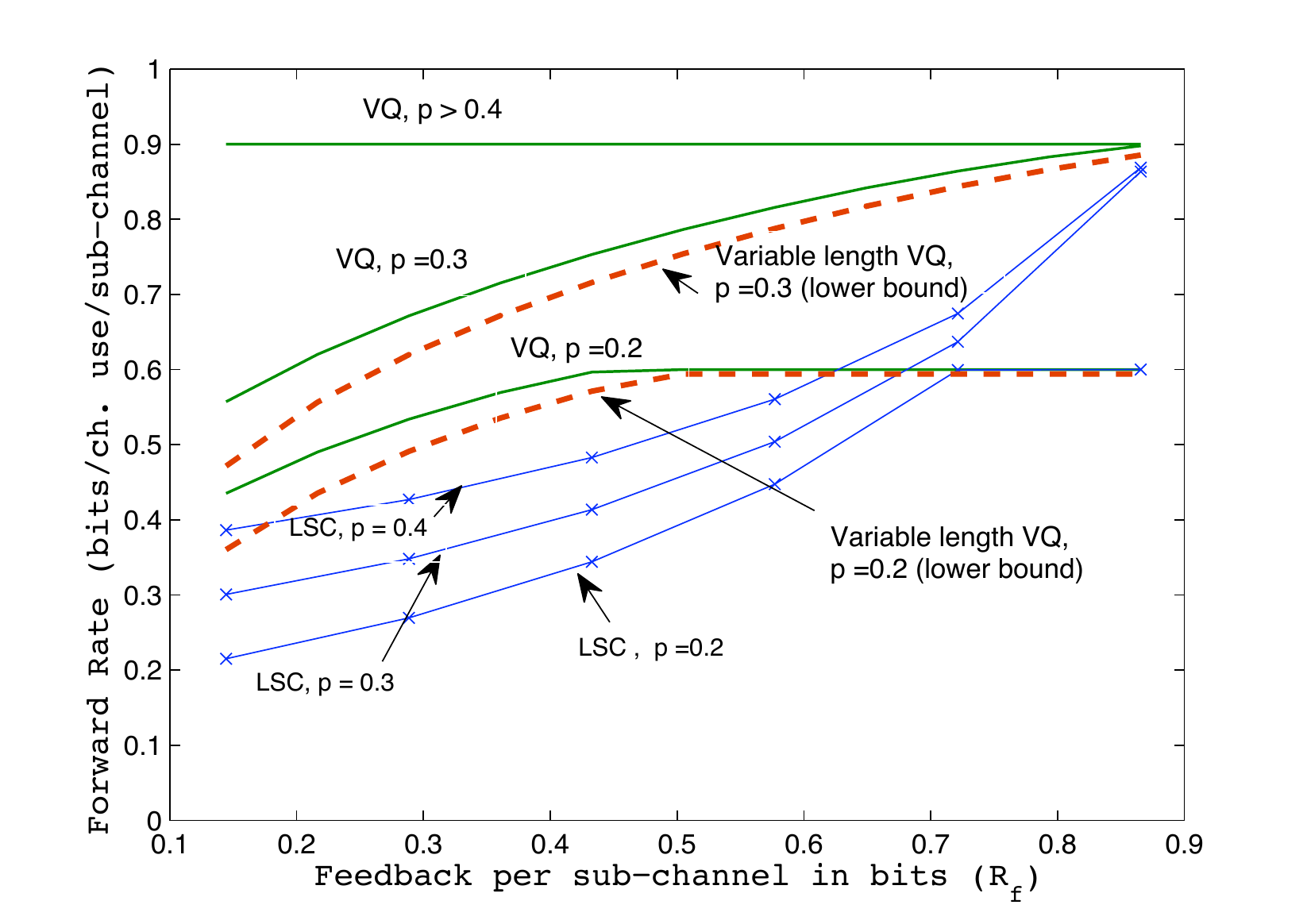}
\caption{Forward rate versus feedback rate for different input 
constraints corresponding to VQ and sub-optimal
feedback schemes using lossless source coding with channel state
reduction (curves labeled with ``LSC''). 
The sub-channels are assumed to be independent.
Other parameters are $q = 0.3$, $C_1 = 3$ and $C_0 =0$.
Also shown is the lower bound on forward rate corresponding to 
$N = 500$ and a variable-length feedback code (dotted lines).}
\label{fig:bincaps}
\end{center}
\end{figure}

The proposition says that in this scenario the forward rate
converges to the upper bound \eqref{eq:C} as $O\left((\log N)/N\right)$.
This is a substantial improvement over the $O\big(\sqrt{(\log N)/{N}}\big)$
convergence rate 
achieved by the fixed-length constant composition feedback 
codes (cf. \eqref{eq:cfix}). 
Although the result in Proposition \ref{prop:bin-var} is stated
for large $N$, a more refined lower bound is
derived in Appendix \ref{ap:prop:bin-var}, which holds
for any finite $N$. 
Fig.~\ref{fig:bincaps} plots a few instances of this lower bound 
against the upper bound \eqref{eq:C} as the feedback rate varies.
%  for $N = 500$, $q = 0.3$,
% $C_1 = 3$, $C_0 =0$ and $p = 0.2, 0.3$.
Clearly, the lower 
bound is fairly close to the optimal forward rate \eqref{eq:C}
and becomes tighter as the feedback rate increases.

\subsection{Practical CSF Codes}
\label{s:codes}

Until now, we have shown that for moderate to large $N$, the rate
distortion trade-off can be approached using random codes.  Such a
code, however, is not practical.
As aforementioned, the Lloyd algorithm can be used to design a
near-optimal vector quantizer for small number of sub-channels
(see \cite{LauLiu04IT, LauLiu04TC, LovHea05TVT, KscLau07IT} and
references therein).  Such a task becomes infeasible with tens or
hundreds of sub-channels, as is the case in many applications.

One practical solution in the case of a large number of sub-channels is
to use a graphical code similar to a 
low-density parity-check (LDPC) code.  Encoding and decoding 
of the source
(channel state vector) are respectively analogous to iterative decoding 
and encoding of a
graphical error-control code.
The complexity of such a code is
in general linear in the number of sub-channels.  For a discussion of
graphical codes for source coding, the reader is referred to
\cite{MarWai06DCC, Wainwr07SPM, CaiSha03Turbo}.
It is more challenging to design and implement variable-length codes.

\subsection{A Sub-optimal Scheme: Lossless Source Coding with Channel State Reduction}
\label{s:lossless}
%\subsection{Coin Tossing: A Simple Sub-optimal Scheme}
%Coin Tossing: Reporting a fraction of the good sub-channels}

For comparison, we also consider a feedback scheme using simple
channel state reduction and lossless source coding in lieu of vector
quantization (henceforth referred to as the ``LSC'' scheme for
convenience).
If the feedback rate is greater than the entropy rate of the channel
state vector, i.e., $R_f > H_2(q)$, then any lossless codes such as the
Huffman code basically suffice.
If the feedback rate is less than the entropy rate,
we consider a simple scheme which reports a fraction $f$ of good
sub-channels, where $f$ is chosen such that the entropy rate $H_2(fq)$
is basically $R_f$.
On average the transmitter is informed of $f q N$ good sub-channels.
The forward rate achieved with this option is 
\begin{equation} \label{eq:Cf}
  C_f = 
  \begin{cases}
    p\,C_1\ , & \textrm{if $ p \le fq$}\\
    fq C_1 + (p - fq)\, \frac{q(1-f) C_1+ (1-q)C_0}{q(1-f)+ (1-q)}\ , &
    \textrm{otherwise}.
  \end{cases}
\end{equation}
%% \begin{equation} \label{eq:Cf}
%% C_f = \left \{ \begin{array}{ll}
%%   pC_1 & \textrm{if $ p \le qf$}\\
%%   qf C_1 + (p - qf)\,C_{av} & \textrm{otherwise},
%% \end{array} \right.
%% \end{equation}
%% where $ C_{av} =  \frac{q(1-f) C_1+ (1-q)C_0}{q(1-f)+ (1-q)}$. 
The expression follows directly from the observation that 
if fewer than fraction $p$
of the sub-channels are reported as good (that is, $fq < p$), the remaining 
$p-fq$ fraction of the sub-channels are chosen at 
random so that the probability of transmitting on a good
sub-channel is given by $(q-fq)/(1-fq)$.
% $\frac{q(1-f)}{q(1-f)+(1-q)}$.

Note that it might be more efficient for the receiver to inform the
transmitter to avoid a subset of bad sub-channels than to report a
subset of good sub-channels, depending on the parameters.  Suppose a
fraction $\bar{f}$ of the bad sub-channels are reported to the
transmitter where $H_2(\bar{f}(1-q)) = R_f$.  The forward rate
achieved with this option is given by
\begin{equation} \label{eq:Cfb}
  \bar{C}_f = 
  \begin{cases}
    p \frac{q C_1+ (1-\bar{f})(1-q)C_0}{q+ (1-\bar{f})(1-q)}\ ,
    & \textrm{if $p < q + (1-\bar{f})(1-q)$} \\
% $p < 1-(1-q)\bar{f}$} \\
    q C_1 + (p - q)C_0\ , & \textrm{otherwise}.
  \end{cases}
\end{equation}
The maximum forward data rate achievable by the LSC scheme is
therefore $\max\{C_f, \bar{C}_f\}$.

\subsection{Numerical Results}
\label{s:2nr}

\begin{figure}
\begin{center}
\includegraphics[width = 5in]{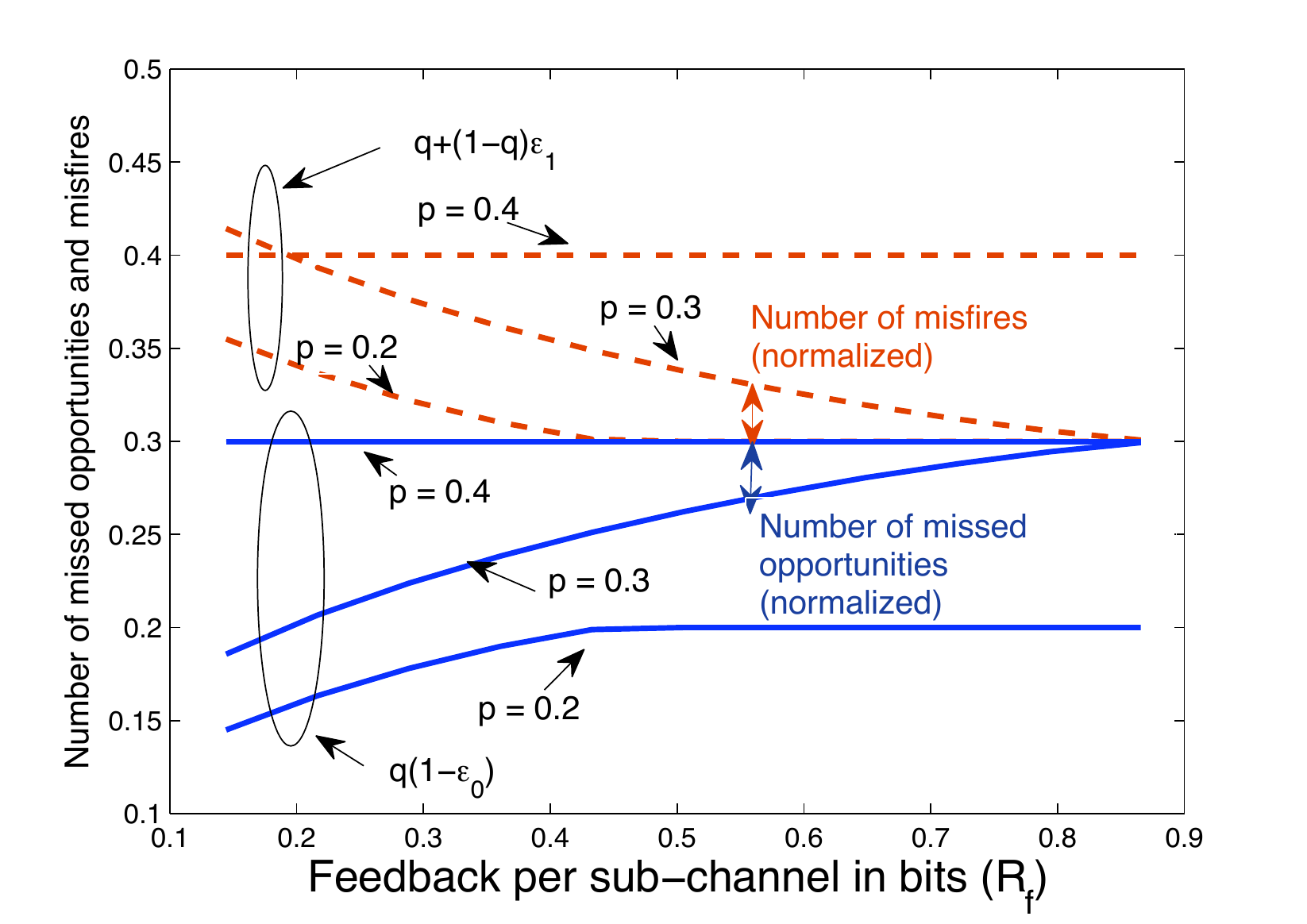}
\caption{The number of misfires and missed opportunities
(normalized by the total number of sub-channels $N$) versus the 
feedback rate for different  input constraints. 
The sub-channels are assumed to be independent.
Other parameters are $q = 0.3$, $C_1 = 3$ and $C_0 =0$.}
\label{fig:e0e1}
\end{center}
\end{figure}

We study the asymptotic performance of the optimal VQ scheme (given by
\eqref{eq:C}) and the sub-optimal LSC scheme (given by \eqref{eq:Cf}
and \eqref{eq:Cfb}).
Fig.~\ref{fig:bincaps} plots the forward rate per sub-channel versus
the feedback rate $R_f$ for different values of the input power
constraint $p$. % We assume $C_1 = 3$, $C_0 =0$ and $q = 0.3$.
The LSC scheme is clearly inferior compared to the asymptotic VQ
scheme with infinite as well as the variable-length VQ scheme at
$N=500$ sub-channels.
At small values of feedback, the VQ scheme gives substantial 
gains (up to $100\%$).
The asymptotic result is quite representative of the performance with
a relatively large number of sub-channels ($N=500$).
As expected, the forward rate increases with the feedback amount and
saturates at $R_f = H_2(q)$, at which point all good sub-channels
can be reported at no loss.

%The gain achieved
%with vector quantization can be understood by studying the 
%corresponding values of mapping probabilities 
%shown in Fig.~\ref{fig:e0e1}. Intuitively, the larger the 
%values of $\eo$ and $\ei$, the more "distortion" or "errors" 
%we allow in the rate distortion feedback code and hence 
%it will require smaller amount of feedback. This is 
%clearly reflected in \eqref{eq:rrfin}. Focusing on 
%case of $p =0.4$ $(> q)$ in Fig.~\ref{fig:e0e1}, $\eo =0$ and 
%$\ei \approx 0.14$,  which implies that we allow
%enough errors of false alarm type so that the required 
%feedback rate is kept small.
%In contrast, we 
%never report bad sub-channel as
%good in the coin tossing scheme thereby incurring extra
%feedback overhead. 

%The reverse holds for $p =0.2$ $(< q) $, where optimal scheme allows
%enough errors of missed opportunity type ($\eo$ is large and $\ei$ is small)
%so that the required feedback rate is again small. Since
%in this case $\ei \approx 0$, bad channels are not reported as
%good .
%%Also note that setting $\ei = 0$ in the rate distortion scheme 
%%corresponds to the case 
%%where we never report a bad channel as good 
%%and further optimize over fraction of good channels that are reported
%%(that is, $1-\eo$). 
%This is same as the coin tossing scheme except that the optimal scheme benefits 
%from vector quantization of state vector as opposed to 
%individual element compression in the coin tossing scheme. 
%%Obviously,
%%the performance will be better than CT and worse than RD
%%scheme. 

The gain achieved
with VQ can be better understood by studying the 
corresponding numbers of missed opportunities and 
misfires shown in Fig.~\ref{fig:e0e1}. Intuitively, the larger the 
values of $\eo$ and $\ei$, the more ``distortion" or ``errors" 
we allow in the rate-distortion feedback code and hence 
it will require smaller amount of feedback. This is 
clearly reflected in \eqref{eq:rrfin}.  Consider the
case of $p =0.4$ $(> q)$ in Fig.~\ref{fig:e0e1}, $\eo =0$ and 
the fraction of misfires $(1-q)\ei = 0.1$,  which implies that we allow
enough misfires so that the required 
feedback rate is kept small.
In contrast, with $p = 0.4$ we 
never report a bad sub-channel as
good for the LSC scheme, and thereby incur extra
feedback overhead. 
The reverse holds for $p =0.2$ $(< q) $, where the optimal scheme allows
enough missed opportunities ($\eo$ is large and $\ei$ is small)
so that the required feedback rate is again small. Since
in this case $\ei \approx 0$, bad channels are not reported as
good.
%Also note that setting $\ei = 0$ in the rate distortion scheme 
%corresponds to the case 
%where we never report a bad channel as good 
%and further optimize over fraction of good channels that are reported
%(that is, $1-\eo$). 
\begin{comment}
This is the same as the LSC scheme except that the optimal
 scheme benefits 
from VQ of the state vector as opposed to 
individual element compression in the LSC scheme. 
\end{comment}

\section{Correlated Two-State Sub-channels}
\label{s:2corr}

In multicarrier systems, the states of the sub-channels are often
correlated.  Consider the same system as in Section \ref{s:2} 
except that the binary (good/bad) channel
states of the $N$ sub-channels, $S_1,S_2,\dots,S_N$ form a stationary
Markov chain.
\begin{comment}  % to appendix???
  Let $\delta_{01} = \Probk{S_i =1|S_{i\pm1} =0}$ and
$\delta_{10} = \Probk{S_i =0|S_{i\pm1} =1}$ for all $i$ denote the
crossover probabilities.  The probability of a sub-channel being good
is then $q=\delta_{01} /(\delta_{01} + \delta_{10})$.
\end{comment}
The optimal feedback scheme is nonetheless a vector quantization
problem, with its asymptotic performance characterized by the
following rate distortion result.

\begin{theorem}\label{th:rdc}
  Given a stationary binary Markov source $\{S_i\}$, a weight constraint on
  every binary reconstruction $w(\hat{\sss}) \leq p$, and the
  single-letter distortion measure $d(s,\hat{s}) = 1_{s>\hat{s}}$, the
  rate distortion function
  % (minimum achievable feedback rate $R_f$) 
  is given by
  \begin{equation} \label{eq:rdc}
    % R_c(p, D) 
    R(D) = \limsup_{N\to \infty} 
    \min_{ P_{\hat{\S}|\S}: \, \substack{
    \sum %_{i=1}^N
    \Prob\{\hat{S}_i=1\} \leq pN \\
    \sum %_{i=1}^N
    \Exp [d(S_i,\hat{S}_i)] \le DN}} \frac1N\, I(\S; \hat{\S})
  \end{equation}     
\end{theorem} 

It is straightforward to prove Theorem \ref{th:rdc} using the same
techniques as developed in \cite{Berger71}, with the additional
weight constraint for the reconstruction. Hence the proof is
omitted. Random codes achieve the rate distortion
function. However, in practice, graphical codes can be designed to
approach the optimal trade-off. In addition, Theorem \ref{th:rdc}
continues to hold if the instantaneous input constraint 
$w(\hat{\s}) \le p$ is replaced by an average input constraint 
$\Exp[w(\hat{\S})] \le p$.

Note that \eqref{eq:rdc} involves minimization over the conditional
distribution of the entire power loading vector, and hence is not a
single-letter characterization of the rate distortion function.
Calculating the rate distortion function for correlated
sources is a hard problem in general. The solution is known only in a few
special cases pertaining to source alphabets, correlation models and
distortion measures \cite{Berger71}.  Even for a symmetric binary Markov
chain and Hamming distortion, the rate distortion function is exactly
known only for very small distortion values \cite{Gray70IT}.
In the CSF problem, the reconstruction $\hat{\S}$ is a binary hidden
Markov process. There is no known close-form expression for
the entropy rate of such processes,
% as a memoryless perturbation of a binary Markov chain, 
although there exist approximations and numerical results in some
cases (see, e.g., \cite{OrdWei04ITW, LuoGuo09IT}
%NaiOrd05ISIT} 
and references therein).  
%%
%% In the following, we obtain simple lower and upper bounds on the
%% mutual information \eqref{eq:rdc}.

Since the exact solution to the optimization problem \eqref{eq:rdc} is
difficult, we will next find upper and
lower bounds on the rate for a given distortion.  Due to the
stationarity of the source, the optimal conditional probability
$P_{\hat{\S}|\S}$ is also stationary.  Consider any stationary
process $\{(S_i,\hat{S}_i)\,|\,i=0,\pm1,\dots\}$ which satisfies the
constraints in \eqref{eq:rdc}, i.e., $\Prob\{\hat{S}_i=1\} \le p$ and
$\Exp [d(S_i,\hat{S}_i)] \le D$.
%% We use $\mathcal{H}(\S)$ to denote the entropy rate of the process $\S$.
Then
\begin{align}
  \liminfty{N} \frac1N I(\S;\hat{\S})
  %% &= \liminfty{N} \frac1N \sum^N_{i=1} I(S_i;\hat{\S}|S^{i-1}_{-\infty}) \\
  &= \liminfty{N} \frac1N \sum^{N-1}_{i=1} I(S_i;\hat{\S}|\S^{i-1}_0) \\
  &= H(S_1|S_0) - \liminfty{N} \frac1N \sum^{N-1}_{i=1} H(S_i|\hat{\S},
  \S^{i-1}_0) \label{eq:NHS} \\
  %%&= I(S_1;\hat{\S}_1^\infty|S_0)
  &\ge H(S_1|S_0) - H(S_1|\hat{S}_1,S_0) \label{eq:HSi} \\
  &= I(S_1;\hat{S}_1|S_0)
  %%&= H(S_1|S_0) - H\big(S_1|\hat{\S},\S^0\big)
\end{align}
where \eqref{eq:HSi} is because of stationarity and because conditioning
decreases the entropy.
% by the mutual information chain rule
% as well as the Markovian property of $\S$.
The rate distortion function can thus be lower bounded as 
\begin{align}
  % \liminfty{N} \frac1N I(\S;\hat{\S}) 
  R(D) \ge \min_{P_{\hat{S}_1|S_1,S_0}} I(S_1;\hat{S}_1|S_0)
\end{align}
with $P_{\hat{S}_1|S_1,S_0}$ satisfying the constraints in
\eqref{eq:rdc}.  The bounding mutual information depends only on the
following four probabilities for the given source:
$q_{s_0s_1} =
  P_{\hat{S}_1,S_0,S_1}(0,s_0,s_1)$ with $s_0,s_1=0 \text{ or } 1$.
%\begin{equation} \label{eq:q01} 
%\end{equation}
% $P_{\hat{S}_1|S_0,S_1}(\hat{s}_1|s_0,s_1)$ for the given source, which
% can be parametrized by four probabilities depending on the values of
% $(s_0,s_1)$.
Denote the lower bound by 
$i_l(q_{00},q_{01},q_{10},q_{11}) = I(S_1;\hat{S}_1|S_0)$.
This bound can be expressed as a function of 
the crossover probabilities denoted by 
$\delta_{01} = \Probk{S_i =1|S_{i\pm1} =0}$ and
$\delta_{10} = \Probk{S_i =0|S_{i\pm1} =1}$ for all $i$.
(The probability of a sub-channel being good
is then $q=\delta_{01} /(\delta_{01} + \delta_{10})$)
An explicit expression for the lower bound is derived in 
Appendix \ref{ap:corr-rd}.
%Consequently, a lower bound of the rate distortion function can be
%obtained as the solution of the following optimization problem:
%\begin{subequations} \label{eq:opts}
%  \begin{align}
%    \text{minimum:} & \quad
%%    \sum_{\hat{s}_1,s_1,s_0}
%%     P_{\hat{S}_1,S_1,S_0} (\hat{s}_1,s_1,s_0) \log \frac{
%%     P_{\hat{S}_1|S_1,S_0} (\hat{s}_1|s_1,s_0) }{
%%     P_{\hat{S}_1|S_0} (\hat{s}_1|s_0) } \\
%    f(q_{00},q_{01},q_{10},q_{11}) = I(S_1;\hat{S}_1|S_0) \\
%    \text{subject to:}
%    &\quad q_{00} + q_{01} + q_{10} + q_{11} \ge 1-p \\
%    &\quad q_{01} + q_{11} \le D \\
%    &\quad 0 \le q_{00}, q_{01}, q_{10}, q_{11} \le 1
%  \end{align}
%\end{subequations}
%which can be easily solved numerically.

In terms of these joint probabilities, the fraction of 
sub-channels that are good, and are correctly 
reported as good is given by $(q-q_{01} - q_{11})$
and the total fraction of sub-channels
reported as good is given by $1- (q_{00} + q_{01} + q_{10} + q_{11})$.
Consequently, an upper bound on the forward achievable rate 
can be
obtained as the solution to the following optimization problem:
\begin{subequations} \label{eq:opts}
  \begin{align}
    \text{maximize:} & \quad
%    \sum_{\hat{s}_1,s_1,s_0}
%     P_{\hat{S}_1,S_1,S_0} (\hat{s}_1,s_1,s_0) \log \frac{
%     P_{\hat{S}_1|S_1,S_0} (\hat{s}_1|s_1,s_0) }{
%     P_{\hat{S}_1|S_0} (\hat{s}_1|s_0) } \\
    (C_1-C_0)(q-q_{01} - q_{11}) + pC_0 \\
    \text{subject to:}
    &\quad 1- (q_{00} + q_{01} + q_{10} + q_{11}) \le p \\
    &\quad i_l(q_{00},q_{01},q_{10},q_{11}) \le R_f \\
    &\quad 0 \le q_{00}, q_{01}, q_{10}, q_{11} \le 1
  \end{align}
\end{subequations}
which can be easily solved numerically.

In order to find an upper bound on the rate distortion function, we
restrict the minimization over $P_{\hat{\S}|\S}$ in
\eqref{eq:rdc} to be a minimization over a finite-dimensional
distribution.  For example, suppose that conditioned on $S_i$ and
$S_{i+1}$, the random variable $\hat{S}_i$ is independent of all the
remaining random variables in $(\S,\hat{\S})$.  By stationarity and the
Markovian property, the joint distribution of $(\S,\hat{\S})$ is
determined by the conditional distribution $P_{\hat{S}_0|S_0,S_1}$.
% $\hat{S}_i$ and
%$\hat{\S}_0^{i-1}$ and $\hat{\S}_{i}^{N-1}$
%and $\hat{\S}_{i+1}^N$ 
%are independent conditioned on $S_i$.
Then it can be shown that, conditioned on $S_{i-1}, S_{i+1},
\hat{S}_{i-1}$ and $\hat{S}_i$, the variable $S_i$ is also independent
of all the remaining random variables in $\S$ and $\hat{\S}$.
Consequently,
\begin{align}
  H(S_i|\hat{\S},\S_0^{i-1})
  &\ge H(S_i|\hat{\S},\S_0^{i-1},\S_{i+1}^N) \\
  &= H(S_i|S_{i-1},S_{i+1},\hat{S}_{i-1},\hat{S}_i).
\end{align}
Substituting in \eqref{eq:NHS},
% Therefore, 
an upper bound for the rate distortion function is obtained
as the solution to the following optimization problem:
\begin{align} \label{eq:ub1}
  R(D) \le \min_{ P_{\hat{S}_0|S_0,S_1} }
  I(S_1;S_2,\hat{S}_0,\hat{S}_1|S_0)
  %%\left[ H(S_1|S_0) - H(S_1|S_0,S_2,\hat{S}_0,\hat{S}_1) \right]
\end{align}
where $P_{\hat{S}_0|S_0,S_1}$ is also subject to the constraints in
\eqref{eq:rdc}.  A simpler but looser bound is obtained if we assume that
$\hat{S}_i$ is independent of everything else conditioned on $S_i$,
for which the mutual information in \eqref{eq:ub1} becomes
 $I(S_1;S_2,\hat{S}_1|S_0)$, and the minimization is over
$P_{\hat{S}_0|S_0}$. Again, let the crossover probabilities be 
given by \eqref{eq:e01}. The upper bound is a function of 
these two crossover probabilities and is denoted by 
$i_u(\eo, \ei) = I(S_1;S_2,\hat{S}_1|S_0)$. An explicit expression
 is derived in Appendix \ref{ap:corr-rd}. 
Consequently, a lower bound on the forward achievable rate 
can be obtained by solving an optimization problem similar to
that in Proposition \ref{prop:RDopt} with the constraint 
\eqref{eq:sp-fb} replaced by $i_u(\eo, \ei) \le R_f$.

It is easily seen that if no reconstruction errors are 
allowed, then both the upper and lower
bounds reduce to the entropy rate of the channel 
state process. In other word, if $H(\hat{\S}|\S) = 0$, or 
equivalently, $\eo= \ei =0$, $q_{10} = q_{00} =1$ and 
$q_{01} = q_{11} =0$, then $i_l(1,0,1,0) = i_u(0,0) = H(S_1|S_0)$.
Next we provide an example
in which the upper bound is not tight. 
Choosing $p=1$ implies that the power
loading vector $\hat{\S}$ can be chosen all ones and that achieves 
the capacity with zero feedback rate. 
However, equivalently choosing $\eo = 0$ and $\ei =1$
gives the upper bound on required feedback rate $i_u(0,1) > 0$. 
%However,  the lower bound $\overline{R}_{c}^{lb} = 0$
%but the upper bound $\overline{R}_{c}^{lb} > 0$.

%Similar to the independent sub-channels case, 
%%the capacity is  given by \eqref{}
%%and $\epsilon_0, \epsilon_1$ are chosen such that 
%%\eqref{} is satisfied. 
%our aim is to maximize the capacity \eqref{eq:dmc-sp-cap} over 
%$\eo$ and $\ei$, subjected to 
%the constraints \eqref{eq:sp-input} and $\overline{R}_{c} \le R_f$. Clearly, using the 
%upper and lower bounds on $\overline{R}_{c}$ will give lower and 
%upper bound, respectively, on the achievable rate (as opposed to the actual capacity
%since using $\overline{R}_c$ instead of $R_c(p, D)$ corresponds to an
%achievable scheme) of RD
% feedback code as a function of $R_f$. 
Fig.~\ref{fig:rrublb} plots the upper and lower bounds 
on achievable forward rate per sub-channel versus $R_f$ for different values of
$q$ with $\dio = 0.3$ and $p = 0.3$. Consider the
plot for $q =0.3$.  There is a
substantial gap between the two bounds for small
feedback rates. However, as the feedback rate increases, the
gap closes and the bounds provide an accurate measure of 
the performance of the VQ scheme with 
correlated sub-channels.
Also shown is the performance of the VQ scheme with independent
sub-channels. Clearly, correlation improves the
forward rate by decreasing the feedback requirement. 

%As for the independent sub-channels, the above described 
%rate distortion code  
%is optimal for correlated sub-channels (limiting ourselves to symmetric random 
%codebooks and Markovian channel model) and 
%any other scheme will only perform worse. 
Later in Section \ref{s:raycorr}, the preceding methodology will be 
utilized to derive bounds on the performance of 
VQ schemes for correlated Rayleigh fading 
sub-channels.
%**analytical insights, plot for different $\delta$, comparison with
%coin tossing.**
%\subsection{Coin Tossing}
%\label{cor-dmc-coin}
%**can again derive vector quantization upper and lower bounds
%and also yakun's (?) individual compression.**
\begin{figure}
\begin{center}
\includegraphics[width = 5in]{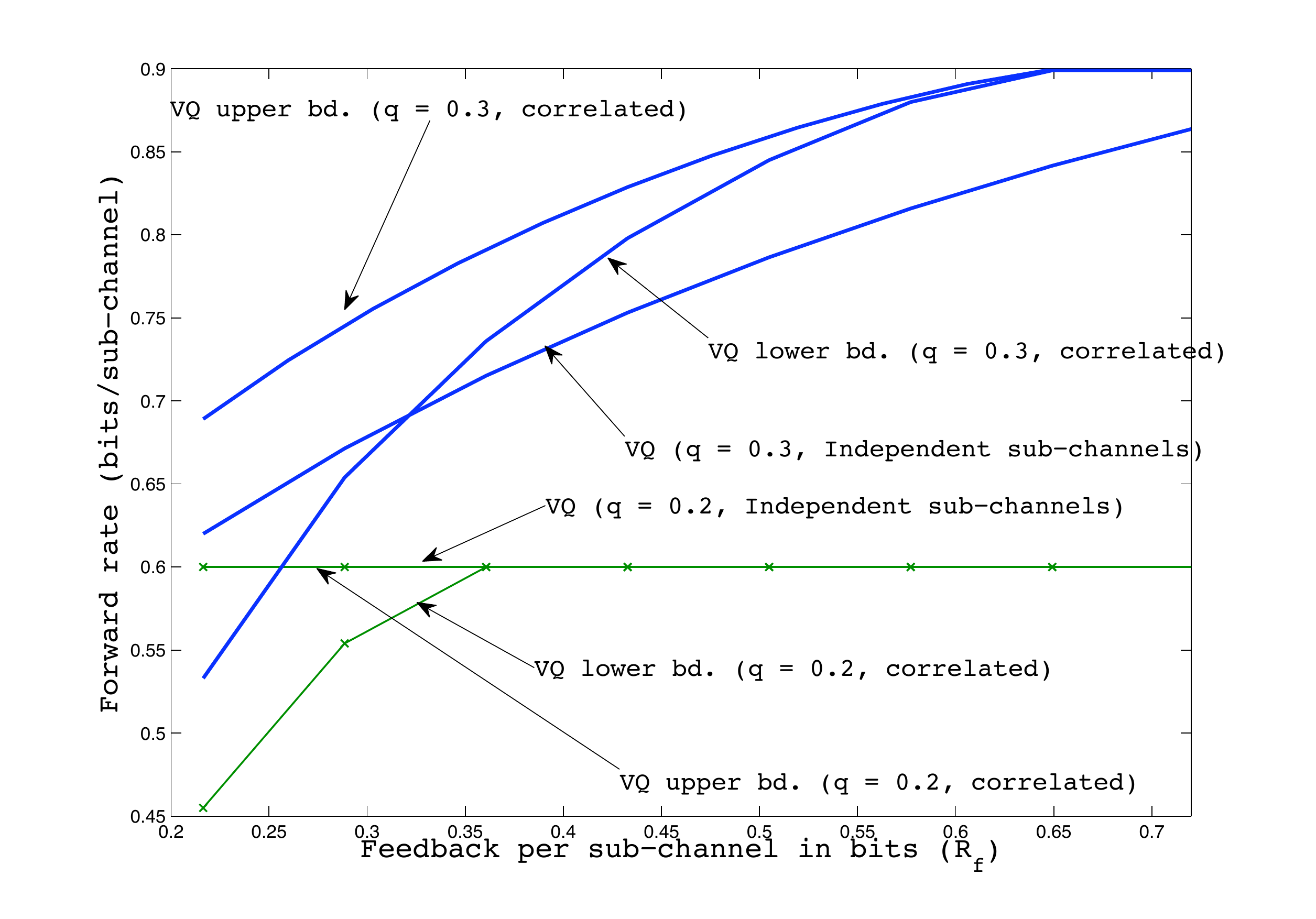}
\caption{Upper and lower bounds on the forward rate with VQ 
versus the feedback rate
for correlated sub-channels. Parameters are  
$\delta_{10}= 0.3, p = 0.3, C_1 =3$ and $C_0 =0$.}
\label{fig:rrublb}
\end{center}
\end{figure}

\section{Independent Rayleigh fading sub-channels}
\label{s:ray}

The design of a limited-rate CSF strategy for
Rayleigh fading sub-channels is again a VQ problem.
Unfortunately, the exact distortion measure, which corresponds to the
capacity-maximizing power loading vectors is difficult to work with
\cite{LauLiu04IT, LauLiu04TC, LovHea05TVT, KscLau07IT}.
In order to simplify the problem, we focus on 
threshold-based schemes, which converts the sequence of 
 Rayleigh fading sub-channels to a sequence of 
``good" (sub-channel gain above the threshold) and ``bad"
(sub-channel gain below the threshold) sub-channels.
This enables the use of the limited feedback schemes developed for 
two-state sub-channels.
It will be shown that, as the number of sub-channels $N \to \infty$, 
the rate achieved with such
a scheme grows at the same rate as that of water-filling
with full channel state information at the transmitter.

The limited feedback problem here differs from the case of two-state 
 sub-channels studied in Sections \ref{s:2} and \ref{s:2corr}
in two key aspects.
First, the threshold which determines the
fraction of sub-channels that are considered good needs to be optimized.
Second, given the total power, the fraction of sub-channels to activate
also influence the amount of powers in each active sub-channel.
 
\subsection{System Model}
\label{ray:mod}
Consider a multicarrier channel with $N$ independent and 
statistically identical,
%ly distributed Rayleigh fading sub-channels, 
where the channel output for the $i$-th sub-channel
is written as
\begin{equation}
  Y_i = H_i X_i + Z_i
\end{equation}
where $H_i$ and $Z_i$ are zero-mean circularly 
symmetric complex Gaussian (CSCG) random
variables. 
Without loss of generality, we assume that 
the channel and noise variance is one, that is, 
$E[|H_i|^2] = E[|Z_i|^2] = 1$. Also, the noise is assumed to
be independent across the sub-channels. 
%where $H_i$ is a circularly symmetric complex Gaussian (CSCG) random
%variable with mean zero and variance $\sigma_h^2$. The additive 
%noise $Z_i$ is CSCG with mean zero and variance $\sigma_z^2$
%and is independent across the sub-channels. 
%so that the $N \times 1$ vector of channel outputs
%across sub-channels is given by
%\begin{equation}
%\y = \H \x + \z
%\end{equation}
%where $\H = \textrm{diag}\left[h_1, h_2, \ldots, h_N\right]$ is the 
%channel matrix in which the diagonal entries are independent,
%circularly symmetric complex Gaussian (CSCG) random
%variables with mean zero and variance $\sigma_h^2$. 
The $N \times 1$ input vector $\X = [X_1, X_2, \ldots, X_N]^\dag$ satisfies 
the average total signal-to-noise ratio (SNR) constraint
$E \left[ \X^\dag \X\right] \le P$.
% \begin{equation}\label{eq:Pcon}
% \end{equation}
%and the $N \times 1$ noise vector $\z$ has CSCG entries with mean zero
%and variance $\sigma_z^2$. 
The channel vector $\H = [H_1, H_2, \ldots, H_N]^\dag$ is assumed
to be known perfectly at the receiver. 
We assume a block fading
model so that $\H$ remains constant for $T$ channel uses and then
changes to an independent value. 
The time dependence is suppressed to simplify notation.

\subsection{Optimal Threshold Based VQ}
\label{ray:sp}

% Further, let $\mu_i = |H_i|^2$
%denote the \emph{sub-channel gain} for the $i$-th sub-channel, which
The gain for the $i$-th  sub-channel, $|H_i|^2$, is 
exponentially distributed with its mean equal to one.
Given a threshold $t \ge 0$, define the $N \times 1$ binary state 
vector $\S$ so that the $i$-th
entry $S_i = 1$ if $|H_i|^2 \ge t$, and $S_i =0$ otherwise.
The probability of a sub-channel being ``good" is denoted as
$q = \Probk{S_i =1} = \Probk{ |H_i|^2 > t } = e^{-t}$.

%Since sub-channel gains are 
%exponentially distributed, the probability of a subchannel being 
%good is given by $q = e^{-t_o/\sh}$. 
% any codebook will have thres, p which can be optimized
%difficulty: thres, p changing with N
%converse 
%achievability: in cap-rate terms?
%compress cap notation
% state optimization prob.
% separate numerical results
Suppose that, on average, the transmitter transmits over or, \emph{activates} 
a fraction $p$  of the sub-channels. The power is distributed
uniformly over the active sub-channels so that each
transmission occurs with SNR equal to $P/(Np)$. 
% Note that, in general, the transmit power can be a function of 
% the number of activated sub-channels during each coherence block. 
Therefore, the expected capacity of a good and bad sub-channel, respectively, is
given by
\begin{equation} \label{eq:ray-cg}
  C_1 = \frac1q \int_{t}^{\infty}  e^{-\tau} \log\left(1 +
    \frac{P\tau}{Np}\right) \intd\tau
\end{equation}
and
\begin{equation} \label{eq:ray-cb}
  C_0 = \frac{1}{1-q} \int_{0}^{t} e^{-\tau} \log
  \left(1 + \frac{P\tau}{Np}\right) \intd\tau.
\end{equation}   

Assume that on average $R_f$ 
bits per sub-channel per coherence block are available for error-free CSF. 
Also, define $B = N R_f$ as the average amount of feedback
summed across all sub-channels.
% at the beginning 
%of every coherence block. 
Similar to the
case of two-state sub-channels, the power loading with 
limited CSF can be seen as a mapping from the space of channels
state vector $\S$ to the space of power loading vectors $\hat{\S}$,
where $\hat{S}_i = 1$ if the $i$-th sub-channel
is activated and $\hat{S}_i = 0$ otherwise. 
%The following proposition gives an upper bound on the performance of the 
%optimal scheme. 

%Note that, here the threshold $t_o$ and hence the 
%statistics $q$ of the \emph{source} $\{S_i\}$ can change with $N$.
%In contrast, the typical asymptotic (large vector length) results in 
%source coding problems are derived assuming that the 
%statistics of the source does
%not change with the length of the vector to be quantized.
%The following bound on achievable distortion (equivalently 
%on the forward achievable rates) can be established.  

We note a key difference between the VQ problem at hand and the usual
stationarity assumption in rate distortion theory: The optimal choice
of the threshold $t$ here may vary with the total number of
sub-channels, hence so does the statistics of the binary source
denoted by probability $q$.
\begin{comment}
Typical asymptotic (large number of dimension $N$) results in 
source coding are derived assuming that the 
statistics of the source do
not change with $N$.
In contrast, here the threshold and hence the 
statistics of the binary source $\{S_i\}$ can change with $N$.
\end{comment}
Nonetheless, the following asymptotic bound on achievable distortion
(equivalently, forward rates) can be established.
   
\begin{proposition}\label{prop:RDopt-Ray}
Given $N$ parallel Rayleigh fading sub-channels and 
an average of $R_f$ bits of feedback per sub-channel per coherence 
block, the following statements hold.
\begin{enumerate}
\renewcommand{\labelenumi}{\alph{enumi})}
\item %\label{i:a}
The forward rate per sub-channel achieved with 
a threshold-based feedback scheme is upper bounded by
$C$, the maximized objective in the following optimization problem:
  \begin{subequations} \label{eq:opt-ray}
  \begin{align}
    \text{maximize:} & \quad C = q(1-\epsilon_0) (C_1-C_0) + p C_0 \\
    \text{subject to:} 
    &\quad H_2(p) - q H_2(\epsilon_0) - (1-q) H_2(\epsilon_1) \le R_f 
    \label{eq:sp-fb1} \\
    & \quad q(1-\epsilon_0) + (1-q)\epsilon_1 = p 
    \label{eq:qep1} \\
    & \quad 0 \le \epsilon_0, \epsilon_1 \le 1
  \end{align}
  \end{subequations}
%{\bf perhaps use $t_0$ in lieu of $q$ throughout this section?}
%in Proposition \ref{prop:RDopt} with $C_1$ given by \eqref{eq:ray-cg}, 
%$C_0$ given by \eqref{eq:ray-cb},
with $q = e^{-t}$, and where the maximization is over $p, t$ and $\eo$.
\item  \label{i:b} 
There exist fixed-length constant composition feedback codes
  which achieve the forward rate per sub-channel given by
  \eqref{eq:cfix} with sufficiently large $N$, where $C$ and $q$ are
  obtained from solving the optimization problem in part~(a).

\item There exist variable-length feedback codes which achieve the
  forward rate per sub-channel given by \eqref{eq:cap-var} with
  sufficiently large $N$, where $C$, $q$, $\eo$ and $\ei$ are obtained
  by solving the optimization problem in part~(a).
\begin{comment}
With fixed-length constant composition feedback codes, 
$\exists N_2 < \infty$ such that for $N \ge N_2$ the forward rate
per sub-channel is lower bounded as
\eqref{eq:cfix}, where $q$ is obtained from the solution to the optimization
problem 
%in point (1) above.
\eqref{eq:opt-ray}. 
%**for large $N$ ?**
%\begin{equation}\label{eq:ray-var-lb}
%C_{fixed} \ge C -  \sqrt{\frac{\log N}{N}} .
%\end{equation} 
\item  \label{i:c}
With variable-length variable composition feedback codes, 
$\exists N_3 < \infty$ such that for $N \ge N_3$ the forward rate 
per sub-channel is lower bounded as
\eqref{eq:cap-var}, where $q$, $\eo$ and $\ei$ are obtained 
from the solution to the optimization problem \eqref{eq:opt-ray}. 
%in point (1).
%\item With variable length feedback codes, the capacity is
%lower bounded as
%\begin{equation}\label{eq:ray-fix-lb}
%C_{var} \ge C -  \frac{\log N}{N} C.
%\end{equation}
\end{comment}
\end{enumerate}
\end{proposition}

Part (a) in Proposition~\ref{prop:RDopt-Ray}
follows directly from the Fano's inequality and is
similar to that of Proposition~\ref{prop:RDopt}.
Note that the upper bound $C$ holds for any finite $N$.
The lower bounds in parts (b) and (c) can be proved
%\eqref{eq:ray-fix-lb} and \eqref{eq:ray-var-lb}
similarly as Propositions~\ref{prop:bin-fix}
and~\ref{prop:bin-var}, respectively,
corresponding to two-state sub-channels. 
The proofs are omitted.
 Although the lower bounds are stated for 
large $N$, more accurate expressions that apply to any 
finite $N$ are derived in 
Appendices \ref{ap:prop:bin-fix} and \ref{ap:prop:bin-var}.
%Obtaining an analytical solution seems to be difficult. However,
%it is straight forward to solve for these bounds numerically.

\begin{figure}
\begin{center}
\includegraphics[width = 5in]{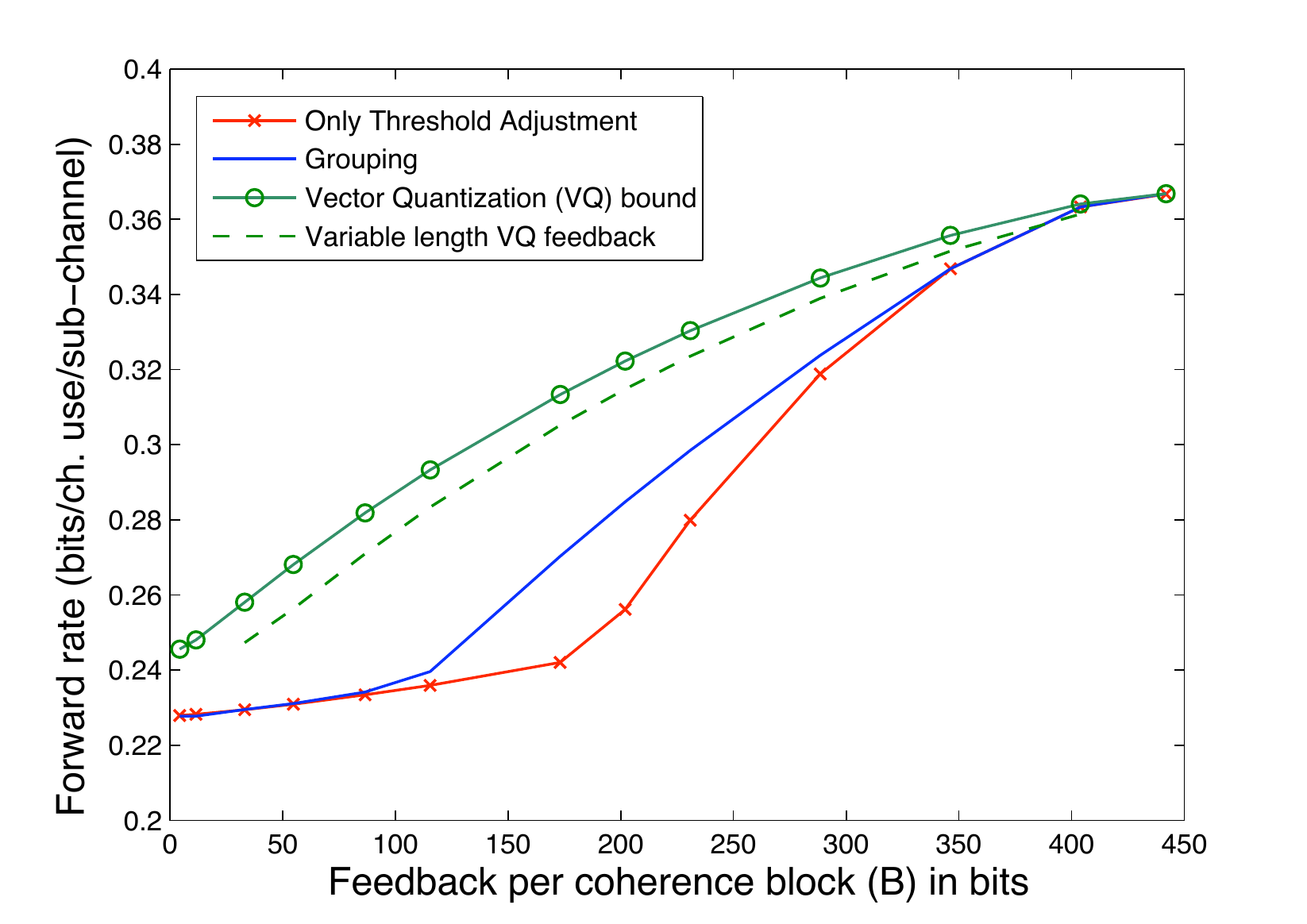}
\caption{Forward rate versus feedback rate for different feedback
  schemes for $N=500$ sub-channels at SNR $P=20$ dB.  For comparison,
  the water-filling capacity with full channel state information at
  receiver is 0.385
  %192.5
  bits per sub-channel use.}
\label{fig:iidraycaps20}
\end{center}
\end{figure}

\begin{figure}
\begin{center}
\includegraphics[width = 5in]{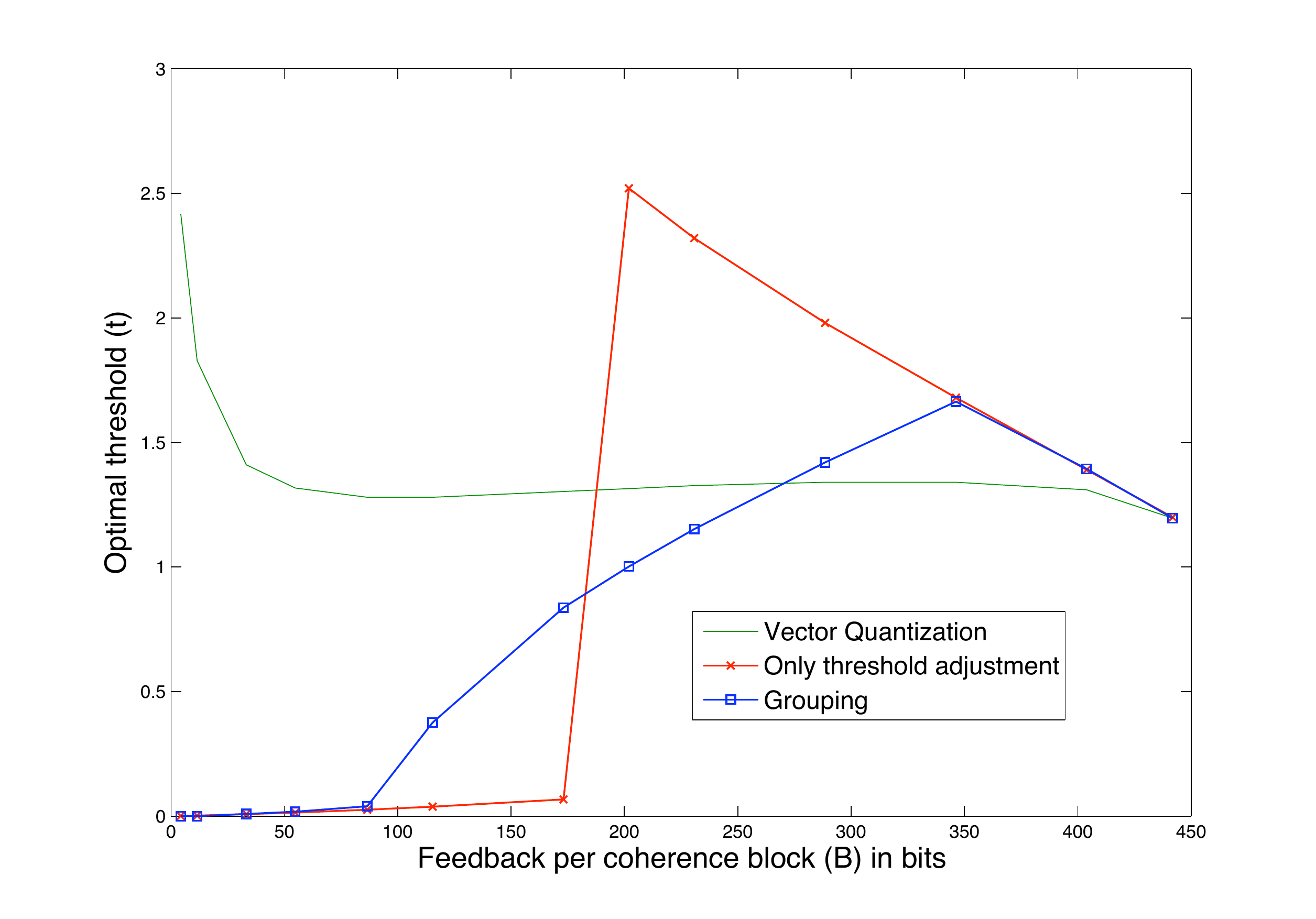}
\caption{Optimal threshold versus feedback for different feedback schemes
with SNR = 20 dB and N = 500.}
\label{fig:iidthres}
\end{center}
\end{figure}  
 
\begin{figure}
\begin{center}
\includegraphics[width = 5in]{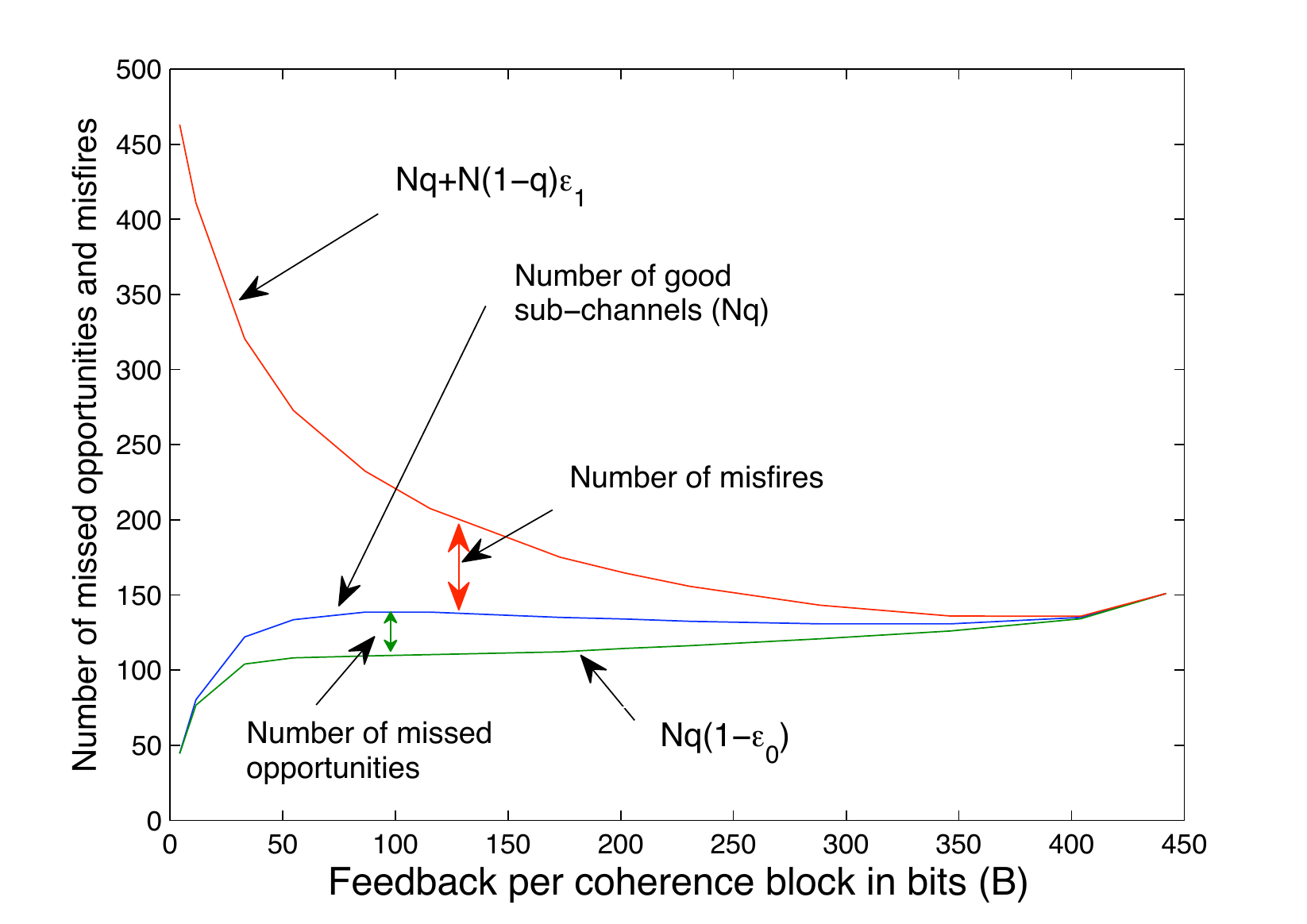}
\caption{Optimal number of sub-channels that 
exceed the threshold together with the optimal number 
of missed opportunities and misfires corresponding to the
VQ feedback scheme for independent 
Rayleigh fading sub-channels. Other parameters are 
 SNR = 20 dB and N = 500.}
\label{fig:map-ray20}
\end{center}
\end{figure}

Interestingly, \eqref{eq:cfix} and \eqref{eq:cap-var} imply
that the rate at which the lower bounds approach 
$C$ depends on the average number of good sub-channels 
$Nq$, which can be much smaller than the total 
number of sub-channels $N$.  
The upper and lower bounds on forward rate versus total 
feedback $B = N R_f$ per coherence block are shown in 
Fig.~\ref{fig:iidraycaps20} for SNR $P = 20$ dB.\footnote{The 
forward rate is measured in bits per sub-channel per channel use whereas 
$B$ is the total number of feedback bits per coherence block. 
Since typically the
coherence block is several hundred channel uses, the
results in Fig.~\ref{fig:iidraycaps20} correspond to
the practical regime in which the feedback rate
is much smaller than the forward data rate.}'\footnote{Numerical 
results for a SNR of 27 dB and $N = 500$ (curves not shown here) show that 
from the case of no feedback to the case of
full feedback (water-filling power allocation) the change in
capacity is merely 16\%. This is consistent with the 
understanding that adaptive power allocation does not
help much at high SNRs. Of course, the gains from the various 
power allocation schemes presented here increase 
with decreasing SNR.}
Only the lower bound corresponding to a
variable-length variable composition feedback code is shown. 
Unless specified otherwise, here and in the subsequent numerical 
results we let $N = 500$. 
%By SNR we refer to $P$. 
The plots show that the upper and lower bounds are
quite close (within 10\%). 
The corresponding optimal threshold is shown in 
Fig.~\ref{fig:iidthres}. The number of {\em good} sub-channels,
number of missed opportunities and the number of 
misfires as a function of amount of feedback 
are shown in Fig.~\ref{fig:map-ray20}. Interestingly, 
for $B \ge 70$ bits, the threshold does not 
change significantly but the number of missed opportunities and
misfires adapt 
to accommodate additional feedback.  
Also shown in Fig.~\ref{fig:iidraycaps20} are curves 
corresponding to other simpler 
sub-optimal feedback schemes to be described in 
subsequent sections. 

%\begin{figure}[htbp]
%\begin{center}
%\includegraphics[width = 4.5in]{e0e1}
%\caption{Optimal mapping probabilities corresponding to the RD
%feedback code for independent Rayleigh fading sub-channels for
%SNR = 5dB.}
%\label{fig:map-ray}
%\end{center}
%\end{figure}

%\begin{figure}[htbp]
%\begin{center}
%\includegraphics[width = 4.5in]{e0e1_20dB}
%\caption{Optimal mapping probabilities corresponding to the RD
%feedback code for independent Rayleigh fading sub-channels for
%SNR = 20dB.}
%\label{fig:map-ray20}
%\end{center}
%\end{figure}

\subsection{Lossless Coding of Feedback with Threshold Adjustment}
\label{ray:coin}

In this section, we consider an alternative feedback scheme using
lossless source coding of reduced channel states.
As in Section \ref{ray:sp}, a threshold $t$ is used to control the
fraction of sub-channels qualifying as good (or bad) in order to meet
the limited feedback constraint.
All 
good (or bad) sub-channels are then reported to the 
transmitter using entropy coding. 
The feedback per coherence block required by this 
scheme is essentially $N H_2(q)$ bits
% ,\footnote{A practical 
% variable-length prefix code typically requires an 
% additional bit per coherence block. We ignore this,
% since a coherence block is likely to contain several 
% hundred channel uses, so that this extra bit contributes
% negligible feedback overhead.} where
with $q = e^{-t}$. 
Note that the feedback constraint 
$H_2(q) \le R_f$ can be met by choosing either $q \le 1/2$ or
$q > 1/2$.
% corresponding to a large and small threshold, respectively. 
The optimal choice corresponds to the one 
that maximizes the forward rate
%The transmission occurs only 
%over the good sub-channels.  In order to satify the power
%constraint \eqref{eq:Pcon}, the power used for each transmission is 
% $P/(qN)$, where .
%Similar to the case of two state sub-channels, we can adopt a
%simple strategy of reporting only a fraction of the good 
%sub-channels (those above the threshold) in order meet
%the feedback constraint. Suppose we report a good sub-channel
%with probability $\epsilon$. In order to satify the power
%constraint \eqref{eq:Pcon}, the transmit power on each active sub-channel is
%$P/(\epsilon qN)$. The capacity therefore is
\begin{equation}\label{eq:ray-coin-cap}
  C(t) = N \int_t^{\infty} e^{-\tau} \log \left(1 + \frac{P\tau}{N q }\right) \intd\tau.
\end{equation}

Fig.~\ref{fig:iidraycaps20} plots the 
forward rate achieved with this scheme versus the feedback 
per coherence block at 20 dB.
For small to moderate feedback, this scheme
performs worse than the optimal VQ scheme
described in the previous section. For higher feedback 
rates, the performance of the two schemes converge.
The optimal threshold versus $B$ for 20 dB is 
shown in Fig.~\ref{fig:iidthres}. 
%Here the values corresponding to 
%only moderate and large amounts of feedback are shown. 
For small amounts of feedback 
the optimal threshold $t$ is close to zero. 
%(not shown in 
%Fig. \ref{fig:iidthres}). 
Furthermore, Figs.~\ref{fig:iidraycaps20} and 
 \ref{fig:iidthres} show that once the 
feedback crosses a certain threshold ($B \approx 170$ bits here), 
the optimal 
threshold decreases and the capacity increases with the amount of 
feedback since more good sub-channels can be reported 
with additional feedback. As $B$ increases further, the
threshold decreases monotonically to an asymptotic value, and 
the capacity reaches its maximum value at around
 $B = B_{max} \approx 440$ bits 
per coherence block. More 
feedback beyond this value cannot be utilized by the threshold-based
scheme.\footnote{The additional 
bits could be used to increase the number of
quantization levels for the power on active sub-channels.
The corresponding increase in rate, relative
to the one-bit quantization assumed here, is typically
quite small \cite{SunHon08IT}.} 

The asymptotic forward rate versus feedback performance 
of this scheme (assuming that the 
amount of feedback and number of sub-channels go to 
infinity) is discussed in \cite{SunHon08IT} and hence
the details are omitted here. A more 
refined analysis 
characterizing the asymptotic growth rate of the maximum
achievable forward rate and $B_{max}$  
as a function of $N$ is given in Section \ref{s:asymp}. 

%\begin{figure}[htbp]
%\begin{center}
%\includegraphics[width = 4.5in]{iidRay_thres}
%\caption{Optimal threshold versus the feedback bits for different feedback schemes.}
%\label{fig:iidthres}
%\end{center}
%\end{figure}    
 
\subsection{Group-Based Power Loading}
\label{ray:group}

Another feedback scheme for comparison is based on sub-channel groups
as well as threshold-based state reduction, which is an enhancement of
the scheme discussed in Section \ref{ray:coin}.  Such a scheme was
originally proposed in \cite{CheBer08JSAC} for reducing feedback in
downlink orthogonal frequency division multiple access (OFDMA)
systems.
The idea is to divide the
sub-channels into $G$ nonoverlapping groups, each 
containing $m=N/G$ consecutive
sub-channels.  
Given a threshold $t$, the receiver informs the transmitter 
to use only those group in which all sub-channel gains exceed $t$.
The probability of this event 
is $e^{-m t}$, so that for large $N$ the average
amount of feedback required per coherence block for this scheme 
can be compressed to the entropy rate $G H_2(e^{-mt})$, 
which should not exceed the feedback constraint $B$.

\subsubsection{Asymptotic Rate Versus Feedback}
\label{ray:group:cap}

Assuming that the transmitter codes across coherence blocks
in frequency and time, the achievable rate is given by
the average mutual information (ergodic capacity),\footnote{
A somewhat more conservative rate is obtained by selecting
the code rate assuming that all active sub-channel 
gains $|h_{gi}|^2= t$ \cite{CheBer08JSAC}.
This does not change the asymptotic results in Section~\ref{ray:group}.}
% \begin{multline}\label{eq:mc-cap}
% C(B) = G \, \int_{
% \begin{array}{c}
% |h_{gi}|^2 \ge t_o,\\
% i=1,\cdots,N_G
% \end{array} }
% \sum_{i=1}^{N_G} \log \left(1+ \frac{P_o}{\sigma_z^2}\,|h_{gi}|^2 \right)
% d \bh_g \\
% = N e^{-(N_G-1)t_o/\sigma_h^2} \int_{t_o}^{\infty} \log \left(1+\frac{P_o}{\sigma_z^2}\,t\right)
% \frac{e^{-t/\sigma_h^2}}{\sigma_h^2} dt
% \end{multline}
\begin{align}\label{eq:mc-cap}
\Cg(m,t) 
%% &= G \, E_{\mathbf{h}_g}\left[ \mathbf{1}_{\{|h_{gi}|^2 \ge t_o, \forall i\}}\sum_{i=1}^{N_G} \log \left(1+ \frac{P_o }{\sigma_z^2}\,|h_{gi}|^2 \right)\right]\\
% &= N e^{-(m-1)t} \int_{t}^{\infty} e^{-\tau} \log \left(1+
&= N q \int_{t}^{\infty} e^{t-\tau} \log \left(1+
  \frac{P\tau}{Nq}\right) \intd\tau.
\end{align}
%% where the indicator function 
%% $\mathbf{1}_{\{|h_{gi}|^2 \ge t_o, \forall i\}}=1$ 
%% if $|h_{gi}|^2 \ge t_o$ for all $i = 1, \ldots, N_G$, and
%% is zero otherwise. 
Note that the rate \eqref{eq:mc-cap} does not depend on the 
coherence block length $T$.
We wish to choose the feedback parameters 
$m$ and $t$ to maximize
$\Cg(m,t)$ subject to the feedback constraint
$G H_2(q) \le B$.
%\eqref{eq:fb-csi}. 
Although it appears to be difficult
to obtain an analytical characterization of the solution for arbitrary $N$,
the  solution 
for large $N$ and $B$ can be characterized as follows.

\begin{proposition} \label{Cap-CSI} %
  For fixed signal-to-noise ratio $P$, as $N\to\infty$ and that $B \to
  \infty$ with $N$, the capacity \eqref{eq:mc-cap} optimized over $t$
  and $m$ is given by\footnote{As $N\to\infty$ and $B\to \infty$ with
    $N$, $o(1)$ is vanishingly small and $O(1)$ is bounded by a finite
    constant.}
  \begin{comment}
We use the following notation. Suppose
$\lim_{N, B\to\infty} \frac{f_1 (N,B)}{f_2 (N,B)} = c$
for functions $f_1(\cdot, \cdot)$ and $f_2(\cdot, \cdot)$.
Then if $c=0$, then we write $f_1 = o(f_2)$,
and if $c > 0$ is a constant, then we write $f_1 = O(f_2)$.
  \end{comment}
\begin{equation} \label{eq:Cg} %
  % \Cg^\star(B) =
  \Cg^\star =
  \begin{cases}
    \sqrt{\frac{P B}{u^\star}} \log(1 +u^\star) + o(1)
    & \textrm{if $ B < B_1$}\\
    \frac{B}{\log N} \log\left(1 + \frac{S \log N}{B} \log
      \left(\frac{N \log N}{B}\right)\right) + O(1)
    & \textrm{if $B_1 \le B < B_{max}$}\\
    P \left(\log N - (1+\eta_2)\log\log N \right) + O(1) & \textrm{if
      $B \ge B_{max}$}
  \end{cases}
\end{equation}
where
%$W = \frac{\log N}{B}$, 
$B_1 = \frac{P}{u^\star} (\log N)^{2-\eta_1}$, $B_{max} = P (\log
N)^{2+\eta_2}$, and $u^\star \approx 3.92$ is the positive solution to
$\log (1+u) = 2\,u/(1+u)$.  In addition, $\eta_1 \in (0,2)$ and
$\eta_2 \in (0,1)$ are functions of $N$ such that $\eta_1 \to 0$ and
$\eta_2 \to 1$ as $N \to \infty$.
\end{proposition}

A sketch of the proof is given in Appendix \ref{ap:prop-grp}.
The appendix also provides optimal threshold and 
group sizes as functions of $B$ and $N$,  
and expressions for $\eta_1$ and 
$\eta_2$.

The capacity 
expressions in Proposition \ref{Cap-CSI} 
are good approximations when 
$N$ is a few hundred and $B$ is a few tens of bits.
%For large enough $N$, the asymptotic inequalities
%in the theorem can be roughly interpreted
%as regular inequalities.\footnote{For this reason we do not write
%the left inequality in \eqref{eq:MaxCsiCon} as $1 \prec B$.}
%Hence in the range of small to moderate feedback $B$ 
%(specifically, $S/u^\star < B \le (S/u^\star) (\log N)^{2-\epsilon_1}$)
In fact, given fixed large enough $N$, the results can be understood
as follows: In the range of relative small to moderate amount of feedback $B$ 
(specifically, $P/u^\star \ll B < (P/u^\star) (\log N)^{2-\eta_1}$)
the capacity is proportional to $\sqrt{B}$.  If $B$ is greater than
$(P/u^\star) (\log N)^{2-\eta_1}$, the capacity increases, 
albeit comparatively slower, with $B$. Finally, the forward rate 
does not increase
when $B$ exceeds $P (\log N)^{2+\eta_2}$. 
The corresponding 
maximum achievable rate is given 
by \eqref{eq:Cg}, which is roughly
$P \log N$. However, the negative second-order (log log) term
in \eqref{eq:Cg} can be substantial, as will be seen in 
the subsequent numerical examples.
A specific numerical example will be provided in Section \ref{ray:group:num}. 

\subsubsection{Optimal Threshold and Group Size}
Expressions for optimal values of the threshold $t$
and group size $m$ as a function of $N$ and $B$
are derived in Appendix \ref{ap:prop-grp}. Here we outline the 
main characteristics of the optimized parameters. 

As expected, when the feedback is in the 
small to moderate range, $B < (P/u^\star) (\log N)^{2-\eta_1}$, 
the optimal group size $m^\star > 1$.
%, that is, the channel state information is \emph{coarse}.
In this range, the threshold increases with feedback and is 
proportional to $\sqrt{B}$. The exact expressions for $m$ and
$t$ are given by \eqref{eq:thres1} and \eqref{eq:NG1}.
For $B \ge (P/u^\star) (\log N)^{2-\eta_1}$, the optimal 
group size $m^\star = 1$. 
Hence, for this range of feedback, the group-based 
scheme reduces to the previous
% only threshold adjustment based 
scheme described in Section \ref{ray:coin}.
Interestingly, here the threshold 
is a decreasing function of $B$ and is given by \eqref{eq:thres2}.
As the feedback increases, decreasing the 
threshold beyond a certain value decreases the capacity.
It is shown in Appendix \ref{ap:prop-grp} that the 
optimal threshold corresponds to feedback $B_{max} = P (\log N)^{2+\eta_2}$ 
and is slightly smaller than $ \log N$.
The exact expression is given by \eqref{eq:optto3}.

% **in the plot:\\
% 1. change y-axix label to "Optimal Parameters and Capacity"\\
% 2. compare with asymptotic expressions\\
% 3. plot rate separately?\\
% 4. different SNRs?
% 5. add plot showing convergence with N**

\subsubsection{Numerical Examples}
\label{ray:group:num}

\begin{figure}
\begin{center}
  \includegraphics[width = 5in]{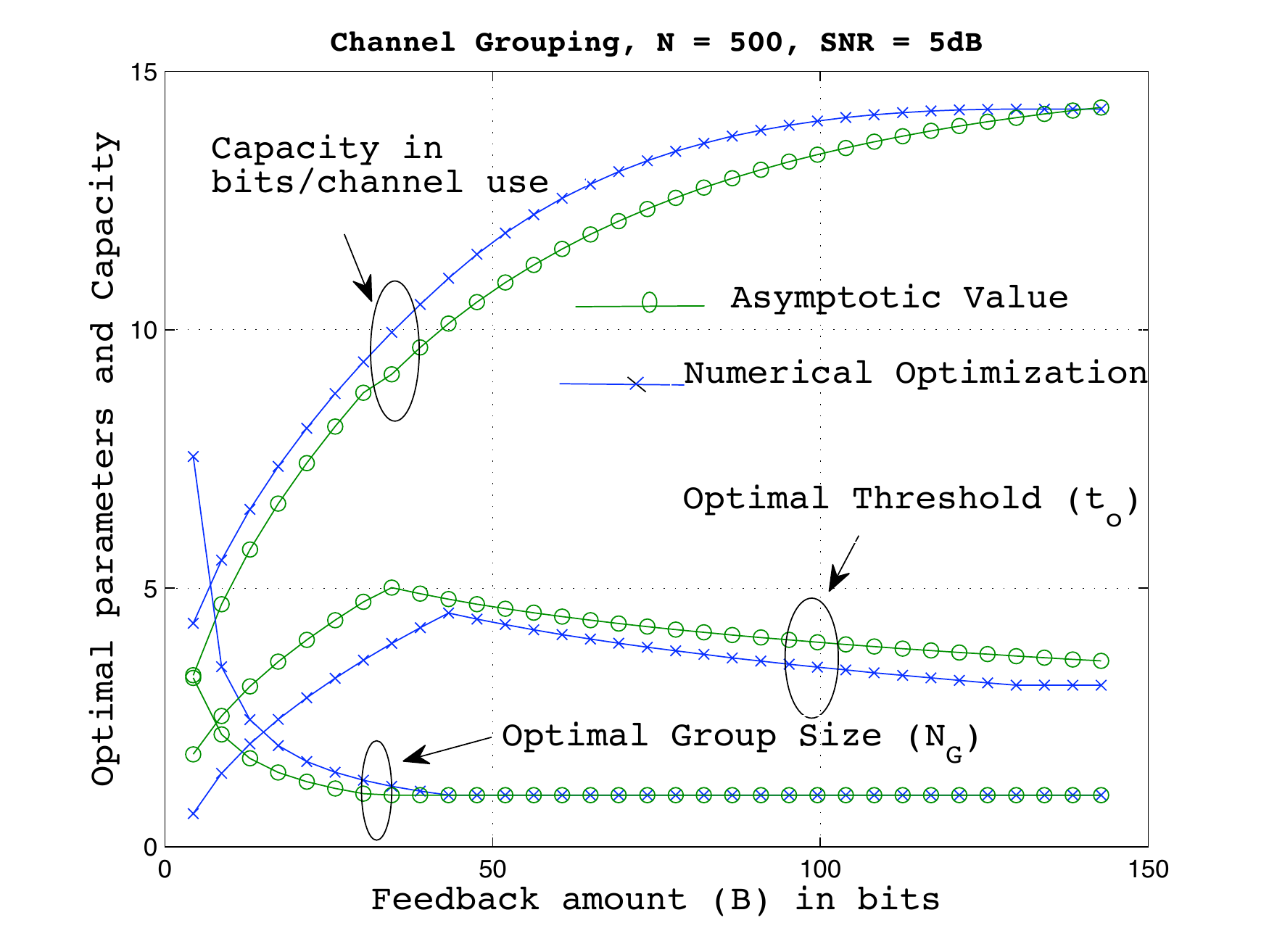}
  \caption{Comparison of numerically optimized values and asymptotic
    values versus feedback per coherence block for $N = 500$ and $P =
    5$ dB SNR.  For comparison, the water-filling capacity is 14.93
    bits per channel use.}
\label{fig:OptPara}
\end{center}
\end{figure}

% Next we present some numerical examples which
% illustrate the preceding asymptotic results.
The preceding asymptotic results are illustrated in
Fig.~\ref{fig:OptPara}.  The optimized
group size, threshold, and corresponding capacity
are plotted as function of the amount of feedback.
These results are obtained by optimizing the 
original capacity expression 
\eqref{eq:mc-cap} subject to the feedback constraint.\footnote{
% $G H_2(q) \le B$.
% \eqref{eq:fb-csi}.
Of course, in practice
$m$ can assume only positive integer values,
as opposed to the real values obtained from the optimization,
which are shown in Fig.~\ref{fig:OptPara}.}
%Unless specified otherwise, for these and subsequent numerical results
%we choose $N = 1000$, $\sigma_h^2 = 1$ and $\sigma_z^2 = 1$.
%
Fig.~\ref{fig:OptPara} also shows the asymptotic analytical results. 
(The values plotted here are refined
versions of the expressions presented in Proposition \ref{Cap-CSI} and
are derived in Appendix \ref{ap:prop-grp}.)
The plot shows that the asymptotic values are close to the 
values obtained from numerical optimization.
As predicted by Proposition \ref{Cap-CSI}, the plot shows that as $B$ increases
from zero, the group size decreases, and the threshold increases.
However, once the group size crosses one
(when $B$ is about 40 bits/coherence block), 
the threshold decreases with $B$ and 
the capacity increases relatively slowly. 
Finally, for large amounts of feedback 
(say, greater than 135 bits/coherence
block for this example) the capacity and threshold saturate,
and increasing the feedback further does not improve performance.
Referring to \eqref{eq:eps1-def}-\eqref{eq:eps2-def},
these values correspond to $\eta_1 \approx \eta_2 \approx 0.25$.

\subsubsection{Performance comparison}
\label{ray:group:comp}
The optimized forward rate \eqref{eq:mc-cap} 
versus $B$ for the group-based scheme %at SNR$\,=20$ dB 
is shown in
Fig.~\ref{fig:iidraycaps20}. 
%We first discuss the results corresponding to SNR $=5dB$.
% compares the performance of grouping scheme with the
%only threshold based scheme and RD scheme. 
At 20 dB SNR, grouping sub-channels can provide about $15\%$ gain over 
threshold adjustment alone 
for small to moderate feedback rates. As the feedback 
increases, the two scheme become the same (since 
the group size converges to $m =1$ at 
around $B \approx 340$ bits). Recall that the 
advantage of the group-based scheme over 
 threshold adjustment alone is limited to 
the feedback range of $B \le (P/u^\star) (\log N)^{2-\eta_1}$.
%Fig. \ref{fig:iidraycaps20} also shows that 
%the vector quantization scheme gives substantial improvement 
%over grouping especially for small to moderate amounts of feedback. 
The optimal threshold for the group-based
scheme is shown in Fig.~\ref{fig:iidthres}.
Observe that the optimal thresholds with and without grouping 
converge for $B \ge 340$. The thresholds behave 
in strikingly different manners for the two schemes when feedback is smaller.

%Fig. \ref{fig:iidraycaps20} compares the performance of the three
%schemes at SNR = 20 dB. 
It is also seen in Fig.~\ref{fig:iidraycaps20} that the VQ scheme 
performs substantially better than both LSC
schemes with and without grouping for small to moderate 
amounts of feedback. 
Namely, VQ saves between 100 and 150 feedback bits per
coherence block over a wide range of target forward rates. 
%Namely, the vector quantization can provide 
%up to $25\%$ gain in the forward rate over grouping scheme.  

%**also fig 5, compare the capacities at 20dB and cap vs fb discussion.**

\subsection{Growth in achievable forward rate}
\label{s:asymp}

In this section we highlight several common features of all three
schemes discussed in Section \ref{s:ray}.
\begin{comment}
Note that as the feedback increases, the VQ 
scheme incurs no errors in reporting the state of 
the sub-channels ($\eo, \ei =0$) and the threshold converges to its 
optimal value.
This is clearly seen in the numerical plots of the optimized threshold
in Fig.~\ref{fig:iidthres}, and of the number of missed opportunities
in Fig.~\ref{fig:map-ray20}. Therefore, 
as seen in Fig.~\ref{fig:iidraycaps20}, the 
performance of the three schemes converge 
as the feedback increases.
%This is because, for sufficiently large feedback, corresponding to
%the RD scheme the optimal 
%$\eo, \ei =0$ (no errors in reporting sub-channels) and 
%the threshold converges to its optimal value given by \eqref{eq:optto2}.
\end{comment}
With sufficient amount of feedback, all three schemes 
\begin{comment}
For all three threshold-based feedback schemes 
considered, the maximum amount
of channel state information is conveyed by setting $m=1$ and optimizing
the threshold as a function of $N$, ignoring the feedback
rate constraint. This   
\end{comment}
correspond to the optimal
``on-off'' power allocation in which the power is
uniformly spread over active channels \cite{SunHon03GLO,SunHon08IT}. 
The optimal threshold is given by 
\begin{equation}
\label{eq:optto3}
t^\star = \left[ \log N - (1+\eta_2)\log\log N -\log P \right] + o(1)
\end{equation}
and the corresponding capacity is given by \eqref{eq:Cg} for $B \ge B_{max}$
(see Appendix \ref{ap:prop-grp}).

From \eqref{eq:eps2-def} in Appendix \ref{ap:prop-grp}, 
as $N\to\infty$, $\eta_2 \to 1$,
so that \eqref{eq:Cg}
states that $O(\log^3 N)$ feedback can achieve
the optimal $O(\log N)$ growth in achievable rate.
This result has been previously presented in \cite{SunHon08IT},
which considers the same threshold-based feedback scheme
considered here without grouping ($m=1$).
%Of course, allowing $N_G>1$ gives more flexibility
%when the feedback rate is small. 
Also, the
numerical examples given in the previous section 
show that for reasonable values of $N$,
the amount of feedback needed to achieve the $O(P \log N)$
forward rate may be closer to $P \log^2 N$ than to $P \log^3 N$.

At the other extreme of small feedback $B \to 0$, 
the three schemes also perform similarly  and the forward rate 
 converges to the SNR $P$ (see Fig.~\ref{fig:iidraycaps20}). 
This limit is 
the ergodic capacity of a Rayleigh fading channel 
without feedback when the bandwidth becomes
large, i.e., $N \to \infty$.
For the VQ scheme, as $B\to 0$, the optimal parameters
converge to either $t \to 0, \eo  \to 0$, or $t \to \infty, \ei \to 1 $,
both of which imply the same transmission strategy. 
Clearly, for threshold adjustment without grouping, the optimal
threshold $t^\star \to 0$ as feedback becomes small.  
Finally, although Proposition \ref{Cap-CSI} does not cover the case
of finite $B$, it is easy to show that for the group-based scheme, as 
the amount of feedback
$B \to 0$, we have $m \to N$, $t \to 0$ and the 
achievable rate $C \to P$.

\section{Correlated Rayleigh fading sub-channels}
\label{s:raycorr}

%**mention that its only an approximation and gives an achievable rate**

In this section, we remove the assumption that the Rayleigh fading
sub-channels are independent.
%and model them as being correlated. 
Suppose the sequence of complex sub-channel coefficients
is a Gauss-Markov process generated by the following first-order autoregressive
model,
\begin{equation}
H_i = \alpha H_{i-1} + \sqrt{1-\alpha^2} \, W_i, \qquad i = 2, \ldots, N,
\end{equation} 
where $W_i$ are i.i.d.\ zero-mean CSCG random 
variables with unit variance, and
 $\alpha \in (0,1)$ represents the correlation between the sub-channels.
Each sub-channel gain $|H_i|^2$ is exponentially
distributed. The sequence of sub-channel gains
$\{|H_i|^2\}$ is a Markov process with joint second order probability
density function given by
\begin{equation}
g(x,y) = \frac{1}{1-\alpha^2}e^{-\frac{x+y}{1-\alpha^2}} I_0\left(\frac{2\alpha \sqrt{xy}}{1-\alpha^2}\right)
\end{equation}
where $I_0(\cdot)$ is the modified Bessel function of the first
kind and zero-order. 

Again, given a threshold $t \ge 0$, the state 
vector $\S$ is defined such that $S_i = 1$ if $|H_i|^2 \ge t$ and 
$S_i =0$ otherwise.  Also, $q = \Probk{S_i=1} = e^{-t}$. 
The sequence $\{S_i\}$ is a \emph{hidden} Markov process rather than i.i.d. 
%and the results corresponding to the 
%two state Markov model of Sec. \ref{s:2corr} cannot be 
%used directly. 
Nonetheless, for a fixed $p$ and $t$, 
the rate-distortion trade-off is still given by~\eqref{eq:rdc}.
However, obtaining upper and lower
bounds on the required feedback rate \eqref{eq:rdc}, 
analogous to the two state
Markov model of Section \ref{s:2corr}, seems to be difficult due to 
the hidden Markov structure of $\S$.

We proceed by assuming that the receiver
approximates the sequence $\{S_i\}$ 
as a first-order Markov chain. 
Ignoring the higher order correlation
in $\S$ results in a larger feedback requirement
(or equivalently, an upper bound on $R(D)$ in \eqref{eq:rdc}).
 Using this upper bound as the 
feedback rate then gives an \emph{achievable} forward rate 
(a lower bound on capacity).
%of course, gives an achievable upper bound on the
%feedback rate (hence an achievable rate or, a lower bound on 
%capacity of the forward channel). 
The transition probabilities for 
the first-order Markov model are, 
\begin{align}
\label{eq:ray-d10}
\delta_{10} &= \Pr \{S_i = 0|S_{i\pm1} = 1\}
=  1-  \frac{1}{q} \int_{t}^\infty \int_{t}^\infty g(x,y) \intd x \intd y,\\
\label{eq:ray-d01}
\delta_{01} &= \Pr \{S_i = 1|S_{i\pm1} = 0\}
=  \frac{q\delta_{10}}{1-q} \,. 
\end{align}

%Furthermore, 
%the sequence $\{s_1, \ldots, s_N\}$ forms a two state Markov chain 
%with transition probabilities given by,
%\begin{align}
%\delta_{10} &=& \Pr \{s_i = 0|s_{i-1} = 1\}
%=  1-  \frac{1}{q} \int_{t_o}^\infty \int_{t_o}^\infty g(x,y)dxdy\\
%\delta_{01} &=& \Pr \{s_i = 1|s_{i-1} = 0\}
%=  \frac{q\delta_{10}}{1-q} 
%\end{align}
%Further let $\hat{\bs}$ and $(\eo, \ei)$ denote the power loading vector
%and mapping probabilities, respectively.
With the first-order Markov approximation, the problem 
 reduces to the 
one discussed in Section \ref{s:2corr}. Consequently, 
%an upper bound on the achievable forward rate 
%is obtained by solving the optimization problem \eqref{eq:opts}
%over variables $p, t, q_{00}, q_{01}, q_{10}, q_{11}$ 
%with  $q, \delta_{10}$ and $\delta_{01}$ given by 
%$e^{-t}$, \eqref{eq:ray-d10} and \eqref{eq:ray-d01}, respectively.
%Furthermore, 
%analogous to Section \ref{s:2corr}, 
an achievable forward rate can be computed by solving the 
optimization problem \eqref{eq:opt-ray} over $p, t, \eo$ and $\ei$
with $q = e^{-t}$ and the constraint \eqref{eq:sp-fb} replaced 
by $i_u(\eo, \ei) \le R_f$. 
An expression for $i_u(\eo,\ei)$ is given in Appendix \ref{ap:corr-rd}.
%Similar to Sec. \ref{dmc:corr}, we consider a random codebook design for 
%feedback code so that, given power loading vector $\hat{\S}$ and 
%mapping probabilities $\eo, \ei$, 
%the required feedback rate per coherence block 
%is $N\overline{R}_c$, where $\overline{R}_c$ is bounded
%from above and below by \eqref{eq:rrub} and 
%\eqref{eq:rrlb}, respectively.
%The capacity for this channels is given by \eqref{eq:cap-ray-sp}.
%Our aim is to maximize the capacity over $\eo, \ei$ and $t_o$
%subjected to the feedback constraint $N \overline{R}_c \le B$.
%Using the upper bound \eqref{eq:rrub} and lower 
%bound \eqref{eq:rrlb} on $\overline{R}_c$  gives the lower and upper 
%bound on forward achievable rate, respectively. 

\begin{figure}
\begin{center}
\includegraphics[width = 5in]{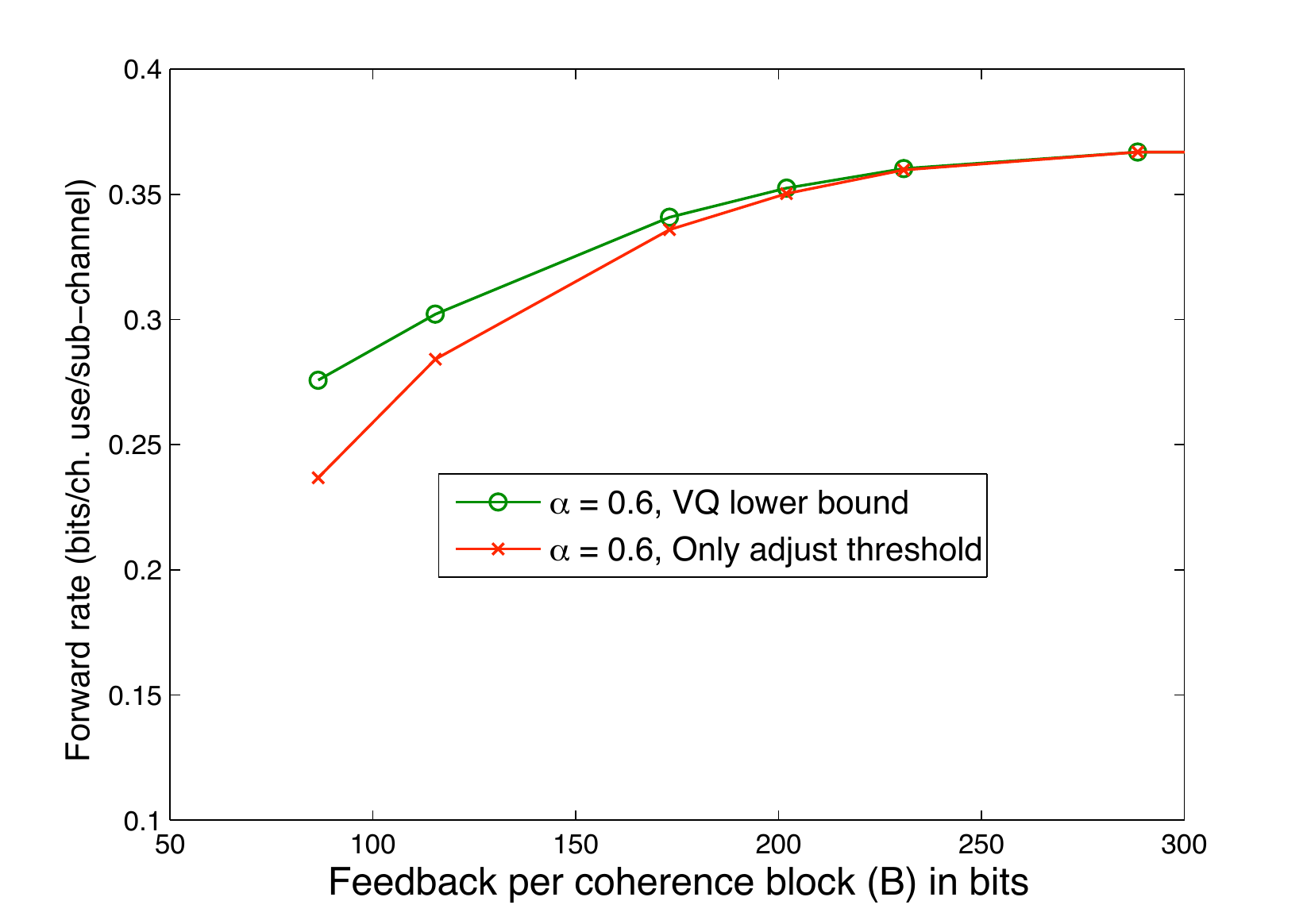}
\caption{Achievable forward rate versus the feedback rate for a VQ
  scheme with correlated fading sub-channels.  Also shown is the
  forward rate achieved with the threshold adjustment feedback scheme
  without grouping.  Other parameters are $N=500$, $P=20$ dB SNR. For
  comparison, the water-filling capacity with full channel state
  information at receiver is
  % 192.5
  0.385 bits per sub-channel use (the same
  as in Fig.~\ref{fig:iidraycaps20}).}
\label{fig:corr-ray-caps20}
\end{center}
\end{figure}

Fig.~\ref{fig:corr-ray-caps20} plots this rate versus 
feedback per coherence block $B = N R_f$ for
$\alpha = 0.6$ and 20 dB SNR.
%For $\alpha = 0.6$, the two bounds are very close 
%and hence provide an accurate estimate of the
%achievable forward rate. In contrast, for $\alpha = 0.9$ (strong 
%correlation), the gap between the bounds is 
%substantial, especially for small feedback rates. 
%Fig. \ref{fig:corr-ray-caps20} plots these bounds 
%for SNR of 20 dB and $\alpha = 0.6$.
Also shown in Fig.~\ref{fig:corr-ray-caps20} 
is the forward rate 
achieved with the simpler LSC feedback scheme that 
encodes the state vector $\S$ and controls the 
amount of feedback by adjusting the threshold
without grouping.
Allowing no errors in the reconstructed channel 
state vector means that the minimum required feedback rate is 
$H(\S)$, which is upper bounded by
$H(S_1|S_0) = q H_2(\delta_{10}) + (1-q) H_2(\delta_{01})$.
%The feedback can now be controlled by adjusting the 
%threshold $t$.
This feedback scheme
 was considered for correlated sub-channels in \cite{SunHon03GLO}. 
%The corresponding capacity is also shown in Fig.~\ref{fig:corr-ray-caps}.
In Fig.~\ref{fig:corr-ray-caps20}, 
VQ performs substantially better 
at small to moderate feedback rates. 
The forward rate values were also computed for $\alpha = 0.9$ (not shown
in Fig.~\ref{fig:corr-ray-caps20}), 
%for $\alpha = 0.9$ (shown only for 5 dB), 
but there the comparison is inconclusive, since the 
lower bound on forward rate achieved with VQ is 
very close to the forward rate achieved with the LSC scheme.

%\begin{equation}
%UB-LB
%\end{equation}
%**explain optimization problem here.**

%**yakun's: $\eo =0$, $\ei =0$ and optimize only over the threshold**
%**coin tossing with vector quantization: $\ei =0$ optimize over 
%thres and $\eo$ -- might reduce to yakun's case. **

%\begin{figure}[htbp]
%\begin{center}
%\includegraphics[width = 4.5in]{corrRay_Caps}
%\caption{Upper and lower bounds on achievable forward rate versus the feedback rate for 
%vector quantization scheme corresponding to correlated fading sub-channels. 
%Also shown is the forward rate achieved with only threshold
%adjustment based feedback. Other parameters are $N = 500$, SNR = 5 dB.}
%\label{fig:corr-ray-caps}
%\end{center}
%\end{figure}

The behavior of the optimal threshold and crossover probabilities $\eo, \ei$ is
the same as for independent sub-channels and hence is not shown here.
Similar to the independent sub-channels case, we can again apply the 
group-based scheme to correlated sub-channels.
Performance evaluation appears to be difficult, but the 
performance of the group-based scheme should be inferior to that of
the VQ scheme and at least as good than without grouping.

\section{Conclusions}
\label{s:con}

%{\bf need to be rewritten to reflect the content - no need to repeat
%  results, perhaps highlight a couple of key points.  pointing to most
%  valuable extensions may be a good idea.}
%two line summary
%other feedback schemes-criterions, other channels
%accounting for training time
%two way channel, fdd instead of tdd

%For a multi-carrier fading channel, we have presented a low complexity power 
%loading scheme that can be matched to small to moderate amounts of feedback. 

We have studied limited feedback of channel states for multicarrier
systems with two-state and Rayleigh sub-channels.  The asymptotic
performance has been characterized, using rate distortion theory, 
along with bounds on the performance loss with 
a finite number of sub-channels.
For Rayleigh channels, the threshold-based VQ scheme shows a substantial
improvement over lossless coding of reduced channel states
based on thresholding,
\begin{comment}
The upper and lower bounds on forward achievable
rate are developed using rate distortion theory. 
It is found that, compared with fixed-length 
feedback codes, variable-length feedback codes 
can have substantially better forward versus feedback 
rate performance. Those results were then applied 
to develop a threshold-based feedback scheme
for Rayleigh fading multicarrier channels. 
Here, at the cost of higher computational complexity, 
threshold-based VQ again gives a substantial
improvement over other suboptimal (lossless) threshold-based
schemes, 
\end{comment}
especially at low to moderate SNRs. 
Of course, this benefit comes at the price of 
higher source coding complexity (e.g., using graphical codes).
%It is also shown that for a large number of Rayleigh fading 
%subchannels $N$, 
%the lossless schemes achieve optimal order 
%growth for the capacity (attained by the waterfilling 
%power allocation) with $S (\log N)^{2+\eta_2}$ feedback bits
%per coherence block, where $\eta_2\in(0,1)$ depends on $N$.
 
%We have analyzed the performance of a limited feedback scheme
%for multi-carrier transmission, which is based on grouping
%the sub-channels. With this scheme the forward rate is an
%increasing function of the amount of feedback, which is
%determined by the sub-channel group size and activation threshold.
%For large $N$ this scheme achieves optimal order growth for the capacity 
%(attained by the waterfilling power allocation) 
%with $S (\log N)^{2+\delta_2}$ feedback bits
%per coherence block, where $\delta_2\in(0,1)$ depends on $N$.
%We have also shown that for this scheme, when the feedback
%rate is relatively small, the optimal feedback overhead
%is one third of the coherence time (again for large $N$).

Our results have assumed perfect channel knowledge
at the receiver and have neglected the feedback 
overhead. In practice, the combined overhead for 
channel estimation and feedback compromises the benefits
of feedback. This becomes more important as the
channel coherence time decreases.
\begin{comment}
Typically the channel is estimated
by means of a training sequence. The feedback shares 
system resources with the forward channel in time-division 
duplex (TDD) or, in frequency-division duplex (FDD) manner.
\end{comment}
A model, which accounts for feedback overhead in 
a multicarrier time-division duplex
system is presented in \cite{AgaGuo08ICC}.
There the overhead is optimized, assuming the lossless
feedback code presented in Section \ref{ray:group}.
Training and feedback overhead in the context of beamforming has
been studied in \cite{SanHon07WCNC,KobCai08ISIT}. That
approach may also be appropriate for the multicarrier 
scenario considered here.

We have also assumed a noiseless feedback link. 
%Considering that the feedback link can be
For the schemes considered here a noisy feedback link 
requires additional overhead
in the form of channel coding or higher transmit power.
Other alternatives include analog CSF
(e.g., see \cite{MarHoc06TSP, CaiJin07X}),
which gives noisy estimates at the transmitter,
and sending a pilot signal 
from receiver to transmitter at the beginning of each 
coherence block\cite{QinBer08JWN} (assuming channel reciprocity applies). 
Comparative advantages and disadvantages of these schemes 
remain to be studied.

%Performance with continuous (as opposed to block) fading channel
%models is also of interest. There the feedback design should take the 
%time correlation into account \cite{YunChu07IWCMC,HuaHea08SP}. 

Finally, the results presented here can conceivably 
be extended to more elaborate system and channel models,
such as continuous fading, instead of block fading
(e.g., see \cite{YunChu07IWCMC,HuaHea08SP}),
and Multiple-Input Multiple-Output (MIMO) OFDM.
%For example, the VQ technique analyzed here can be applied to other
%types of channels and systems, 
%systems which have a large number of channel coefficients. 
Such systems typically operate at low
SNRs per antenna (or coefficient)
and hence can benefit substantially from adaptive 
power loading\cite{KhaMon07TWC, HooCai08TCOM}. 
Limited CSF schemes for downlink and uplink OFDMA
are presented in \cite{AgaMaj08JSAC, CheBer08JSAC, AgaHon07Crown}
and the references therein. (See also the comprehensive
survey of the limited feedback literature in \cite{LovHea08JSAC}.)
Extensions of the VQ scheme presented here
to those settings is also left for future work.

\appendices

\section{Proof of Proposition \ref{prop:bin-fix}}
\label{ap:prop:bin-fix}

First we introduce some notation, then describe the construction of
a fixed-length, constant-composition feedback code, and
finally give the performance bounds as a function of $N$.

% \emph{Notation}:\\
Let $\mathcal{P}_{\s}(1)$ represent the \emph{empirical probability}
of ones in a length $N$ binary vector channel state $\s$. Namely, 
$\mathcal{P}_{\s}(1) = q$ implies that
%\footnote{
%For simplicity, hereon we will assume that empirical probabilities and 
%$N$ are such that the flooring operation is not required.}
%\begin{equation}
%\sum_{i = 1}^{N} s_i = \floor{qN}.
%\end{equation}
$\sum_{i = 1}^{N} s_i = qN$.
A similar definition holds for $\mathcal{P}_{\hat{\s}}(1)$. 
Furthermore, for any pair of vectors $(\hat{\s}, \s)$, 
 $\mathcal{P}_{\hat{\s}|\s}(0|1)$ represents the
empirical probability of zeroes in $\hat{\s}$ at the positions
where $\s$ has ones. In other words, 
$\mathcal{P}_{\hat{\s}|\s}(0|1) = \eo$ implies that
$\sum_{i = 1}^{N} 1_{\{\hat{s}_i =0, s_i =1\}} = \eo \mathcal{P}_{\s}(1) N$.
% \begin{equation}
% \sum_{i = 1}^{N} 1_{\{\hat{s}_i =0, s_i =1\}} = \eo \sum_{i = 1}^{N} 1_{\{s_i =1\}}  = \eo \mathcal{P}_{\s}(1) N .
% \end{equation}
Let $\mathbf{\mathcal{T}}_{\S;q}$ represent the set of all 
constant composition length $N$ channel state vectors $\s$
which have $qN$ number of ones, that is, 
$\mathbf{\mathcal{T}}_{\S;q} = \{\s : \mathcal{P}_{\s}(1) = q\}$. 
Similarly, define the set of constant composition 
power loading vectors 
$\mathbf{\mathcal{T}}_{\hat{\S};p} = \{\hat{\s} : \mathcal{P}_{\hat{\s}}(1) = p\}$.

%\emph{Feedback Code:}\\
%Let  $\mathbf{\mathcal{T}}_{\hat{\S}|\S}$ be the smallest 
%subset of $\mathbf{\mathcal{T}}_{\hat{\S};p}$ such that for 
%every vector $\S$ in $\mathbf{\mathcal{T}}_{\S;q}$ there exists
%$\hat{\S} \in \mathbf{\mathcal{T}}_{\hat{\S}|\S}$ such that 
%$\mathcal{P}_{\hat{\S}|\S}(0|1) = \eo$ and 
%$\mathcal{P}_{\hat{\S}|\S}(1|0) = \ei$. 
%The set $\mathbf{\mathcal{T}}_{\hat{\S}|\S}$ forms the \emph{feedback
%codebook}. 
%Namely, every channel realization vector is matched to the 
%closest (in hamming distance) typical vector $\S \in \mathbf{\mathcal{T}}_{\S;q}$
%and the index of the corresponding power loading vector in the
%set $\mathbf{\mathcal{T}}_{\hat{\S}|\S}$ is fed back.
%Therefore, the amount of feedback required per coherence block is 
%$\log_2|\mathbf{\mathcal{T}}_{\hat{\S}|\S}|$ bits.
Let the {\em feedback codebook} be a subset
$\mathbf{\mathcal{T}}_{\hat{\S}|\S} \subset \mathbf{\mathcal{T}}_{\hat{\S};p}$.
%be the \emph{feedback codebook}. 
The amount of feedback required per coherence block is therefore
$\log_2|\mathbf{\mathcal{T}}_{\hat{\S}|\S}|$ bits.
A vector $\s \in \mathbf{\mathcal{T}}_{\hat{\S}|\S}$ is said to
\emph{cover} $\s \in \mathbf{\mathcal{T}}_{\S;q}$ if 
$\mathcal{P}_{\hat{\s}|\s}(0|1) = \eo$ and 
$\mathcal{P}_{\hat{\s}|\s}(1|0) = \ei$. Note that 
depending on the size of the codebook, all the vectors in 
$\mathbf{\mathcal{T}}_{\S;q}$ might not be covered. 
The feedback occurs as 
follows: Every channel realization vector is matched to the 
closest (in hamming distance) typical vector 
$\s \in \mathbf{\mathcal{T}}_{\S;q}$. The receiver then looks 
up a $\hat{\s}$ in $\mathbf{\mathcal{T}}_{\hat{\S}|\S}$ 
which covers $\s$ feeds back the index of $\hat{\s}$. 
If there is no such $\hat{\s}$ then a random index 
is fed back. Therefore, the elements of the codebook 
$\mathbf{\mathcal{T}}_{\hat{\S}|\S}$ should be chosen carefully
to minimize the distortion or, equivalently, maximize the forward rate.
Next, along the lines of \cite{Goblic62}, we will evaluate the 
average performance assuming that the subset $\mathbf{\mathcal{T}}_{\hat{\S}|\S}$
is chosen at random. Then by usual argument we can claim that 
there must exist a structured fixed-length constant composition 
codebook that does at least as well.

Consider the performance of the codebook as a function of $N$.
The Type covering lemma \cite{CsiKor81} suggests that 
we can find a codebook with size 
$|\mathbf{\mathcal{T}}_{\hat{\S}|\S}| = g(N) 2^{i(\eo, \ei)N}$,
where $g(N)$ is a polynomial, such that every vector in
$\mathbf{\mathcal{T}}_{\S;q}$ is covered. Here, 
 we use random coding arguments to get an estimate of
 $g(N)$. Suppose $\mathbf{\mathcal{T}}_{\hat{\S}|\S}$ consists of
 $M$ vectors that are randomly and independently\footnote{We
have $N\choose{Np}$ choices in each drawing.} drawn from 
$\mathbf{\mathcal{T}}_{\hat{\S};p}$.  

Given a vector $\s \in \mathbf{\mathcal{T}}_{\S;q}$, probability that 
it will be not be \emph{covered} by the $M$ randomly 
chosen vectors is given by 
\begin{align}
p_{n} &= [1-p_c]^M %= e^{M \log (1-p_c)}\\
\label{eq:pncineq}
%&
\le e^{-Mp_c}
\end{align}
where 
\begin{equation}
p_c =  \frac{{qN \choose \eo qN} {(1-q)N \choose \ei (1-q)N}} {{N\choose pN}}.
\end{equation}
%\begin{eqnarray}
%p_{n} &=& \Pr\{\not \exists \hat{\S} \in \mathbf{\mathcal{T}}_{\hat{\S}|\S}
%\mathcal{P}_{\hat{\S}|\S}(0|1) = \eo, \,\mathcal{P}_{\hat{\S}|\S}(1|0) = \ei\}\\
%&=& [1-p_c]^M\\
%&=& e^{M \log (1-p_c)}\\
%\label{eq:pncineq}
%&\le& e^{-Mp_c}
%\end{eqnarray}
%where 
%\begin{equation}
%p_c =  \frac{{qN \choose \eo qN} {(1-q)N \choose \ei (1-q)N}} {{N\choose pN}}.
%\end{equation}
%Using the Type covering lemma \cite{CsiKor81}, the size of the 
%code book is given by
%\begin{equation}
%|\mathbf{\mathcal{T}}_{\hat{\S}|\S}| =  \frac{{N\choose pN}}{{qN \choose \eo qN} {(1-q)N \choose \ei (1-q)N}}.
%\end{equation}
Further using Robbin's approximation \cite{RobbinAMM55, CsiKor81} for the factorial
\begin{equation}\label{eq:rob}
\sqrt{2\pi}\, n^{n+\frac{1}{2}} e^{-n + \frac{1}{12(n+1)}} \le n! \le \sqrt{2\pi} \, n^{n+\frac{1}{2}} e^{-n + \frac{1}{12n}},
\qquad n\ge 1,
\end{equation}
it is straightforward to show that
\begin{equation}\label{eq:pcineq}
p_c \ge N^{-1/2} 2^{-K_2}\,2^{-i(\eo,\ei) N}
\end{equation}
where the mutual information $i(\eo,\ei)$ is given by \eqref{eq:rrfin} 
and $K_2 = \log_2(\sqrt{2\pi}e^{5/12}) - \frac{1}{2}\log_2(p(1-p))$ is 
a constant. Now choosing $M = (\log N) N^{1/2} 2^{K_2}\,2^{i(\eo,\ei) N}$
and combining \eqref{eq:pncineq} and \eqref{eq:pcineq} gives
$p_{n} \le 1/N$. 
Therefore, with this codebook the feedback 
rate per sub-channel per coherence block given by 
\begin{equation} \label{eq:IfiniteN}
 \frac{\log_2|\mathbf{\mathcal{T}}_{\hat{\S}|\S}|}{N} = \frac{\log_2M}{N} = 
i(\eo,\ei) + \frac{(\log_2N + 2K_2 + 2\log_2(\log N))}{2N}
\end{equation} 
converges to the mutual information $i(\eo,\ei)$. 
%, we require up to $\frac{1}{2} \log_2 N  + K_1$
%additional bits per coherence block in order  

Next we bound the average forward rate achieved by this codebook. 
%the fixed length 
%feedback code using at most $\log_2M$ bits (given by \eqref{eq:IfiniteN}) 
%per coherence block.
For this we need to account for the variations
 in the channel gain vector $\S$. Consider the set of 
channel gain vectors with fraction of ones in the range
$(q(1-\epsilon), q(1+ \epsilon))$, that is, say
$\mathbf{\mathcal{T}}_{\S;q}^\epsilon = \{\s: \mathcal{P}_\s(1) \in (q(1-\epsilon), q(1+\epsilon))\}$.
Using Chernoff's inequality \cite{ChernoAMS52}, the following is easily seen
\begin{equation}\label{eq:cher-bd}
\Pr\bigg\{\bigg|\sum_{i = 1}^N S_i - qN \bigg| \ge q\epsilon N\bigg\} \le p_u = 2 \exp \left(-\frac{qN \epsilon^2}{4} \right),
\qquad 0 \le \epsilon \le 2(1-q), 
\end{equation}
which in turn implies that 
%\begin{equation}\label{eq:typ-bd}
%\Pr\{\mathbf{\mathcal{T}}_{\S;q}^\epsilon\} \ge 1-  2 \exp \left(-\frac{qN \epsilon^2}{4} \right), \qquad 0 \le \epsilon \le 2(1-q).
%\end{equation}
$\Pr\{\mathbf{\mathcal{T}}_{\S;q}^\epsilon\} \ge 1-  p_u$.
% \begin{equation}\label{eq:typ-bd}
% \Pr\{\mathbf{\mathcal{T}}_{\S;q}^\epsilon\} \ge 1-  p_u .
% \end{equation}
Using the definition of $\epsilon$ above, the ergodic
capacity is lower bounded as
\begin{align}\label{eq:fix-lb-temp}
C_{fixed} &\ge q(1 - \eo - \epsilon)(C_1 - C_0)\left(1-  p_u \right) (1-p_n) + p C_0\\
& \ge C - C \left(\epsilon + p_u + \frac{1}{N}\right),
\end{align}
where $C = q(1 - \eo)(C_1 - C_0) + pC_0$. The loss factors $1-p_n$ 
and $1-p_u$ in \eqref{eq:fix-lb-temp} account for the fact that channel gain 
vectors might not be covered or, might not be in the set 
$\mathbf{\mathcal{T}}_{\S;q}^\epsilon$, respectively. Further, choosing 
$\epsilon = \sqrt{\frac{2\log (Nq)}{Nq}}$ gives a tight lower bound as
\begin{equation}\label{eq:cfix-final1}
C_{fixed} \ge C - C \left(\frac{\sqrt{2\log (Nq)}+2}{\sqrt{Nq}} + \frac{1}{N}\right) \ .
\end{equation}
%Further, allowing an additional distortion of $\delta$ to force the number of 
%feedback per sub-channel to be $I(S; \hat{S})$ gives
%\begin{equation}
%H_2(p) - qH_2(\eo+\delta) - (1-q)H_2\left(\ei + \frac{q\delta}{1-q}\right) +
% \frac{1}{N} \left(\frac12 \log_2N + K_1\right) = I(S;\hat{S})
%\end{equation}
%which gives
%\begin{equation}
%\delta \le \frac{1}{N K_2} \left(\frac12 \log_2N + K_1\right)
%\end{equation}
%where $K_2 = q[H_2'(\eo)+ H_2'(\ei)]$. 
%
%Replacing $\eo$ by $\eo+\delta$
%in the ergodic capacity expression \eqref{} gives a further lower 
%bound as 
%\begin{align}
%C_{fixed} &\ge q(1 - \eo- \delta - \epsilon)(C_1 - C_0)\left(1-  2 \exp \left(-\frac{qN \epsilon^2}{4} \right) \right) + p C_0\\
%& \ge C - C \left(\delta + \epsilon +  2 \exp \left(-\frac{qN \epsilon^2}{4} \right) \right),
%\end{align}
%Finally, maximizing \eqref{eq:fix-lb-temp} over $\epsilon$ 
%gives the lower bound \eqref{eq:cfix}.
In summary, since \eqref{eq:cfix-final1}  bounds the average performance of 
a randomly generated codebook with $M$ codewords, there
must exist a codebook of size $M$ which performs at least
as good as this lower bound.   

Lastly, we consider the convergence rate to the upper bound given by
Proposition \ref{prop:RDopt}.
Note that required feedback rate \eqref{eq:IfiniteN}
is more than mutual information $i(\eo, \ei)$.
If  we do not wish to allow the feedback to exceed $i(\eo, \ei)$ bits per
sub-channel,  
an additional distortion of $\delta$ can be introduced so that 
$\eo$ and $\ei$ are replaced by $\eo +\delta$
and $\ei+ \delta'$, respectively, where $\delta' = {q\delta}/(1-q)$
(so that the input power constraint \eqref{eq:qep} is satisfied). 
Using Taylor series expansion, for a small 
enough $\delta$ we have 
%\begin{multline}
%q H_2(\eo) + (1-q) H_2(\ei) + \frac12 q[H_2'(\eo) + H_2'(\ei)] \delta \le q H_2(\eo + \delta)
% + (1-q) H_2(\ei + \frac{q\delta}{1-q})\\
% \le q H_2(\eo) + (1-q) H_2(\ei) + 2 q[H_2'(\eo) + H_2'(\ei)] \delta
%\end{multline}
\begin{equation}\label{eq:temp-ineq2}
 - 2 q[H_2'(\eo) + H_2'(\ei)] \delta \le i(\eo+\delta, \ei + \delta' ) -  i(\eo, \ei)
 \le - \frac12 q[H_2'(\eo) + H_2'(\ei)] \delta.
\end{equation}
Substituting  $\eo +\delta$ and $\ei+ \delta'$ for $\eo$ and 
$\ei$ in \eqref{eq:IfiniteN} and
using \eqref{eq:temp-ineq2}, for large enough $N$ and 
small enough $\delta$ we have
\begin{equation}
 \frac{\log_2|\mathbf{\mathcal{T}}_{\hat{\S}|\S}|}{N} \le i(\eo,\ei)  - \frac12 q[H_2'(\eo) + H_2'(\ei)] \delta + \frac{\log_2 N}{N} \ .
\end{equation}

Now choosing $\delta = K_1 \frac{\log_2 N}{N}$ where $K_1 = \frac{2}{q[H_2'(\eo) + H_2'(\ei)]}$,
gives that the number of feedback bits per sub-channels 
$\frac{\log_2|\mathbf{\mathcal{T}}_{\hat{\S}|\S}|}{N} \le i(\eo, \ei)$. 
The additional distortion of $\delta$ will result in 
a loss in capacity which can be quantified by replacing 
$\eo$ by $\eo+\delta$ in \eqref{eq:fix-lb-temp} which gives a lower
bound on forward rate as
\begin{equation}\label{eq:cfix-final2}
C_{fixed} \ge C - C \left(\frac{\sqrt{2\log (Nq)}+2}{\sqrt{Nq}} + \frac{1}{N} + K_1 \frac{\log_2 N}{N} \right)
\end{equation}
which for large enough $N$ can be further lower bounded as
\eqref{eq:cfix}.

%% \begin{equation}
%% C_{fixed} \ge C - C \left(\frac{2\sqrt{\log (Nq)}}{\sqrt{Nq}} \right).
%% \end{equation}

\section{Proof of Proposition \ref{prop:bin-var}}
\label{ap:prop:bin-var}
%The proof primarily follows the technique given in \cite{Pinkst67}
%albeit with slight modifications as listed next.
%The main difference lies in incorporating the average input 
%power constraint and avoiding the use of \emph{partitioning 
%functions} by exploiting special structure of the 
%problem at hand. 

We again  resort to the random coding techniques as in \cite{Pinkst67}, 
assuming that corresponding to each channel state vector $\s$, 
the power loading vector $\hat{\s}$ is produced with 
Bernoulli-$p$ distribution.
%\emph{conditional} distribution given by $\Pr\{\hat{S}_i = 1|S_i =0\} = \ei$
%and $\Pr\{\hat{S}_i = 0|S_i =1\} = \eo$. Note that 
%as opposed to the usual random coding arguments in
%which the random codewords are produced with the 
%same marginal distribution, here the marginal of the codeword
%depends on the composition of the corresponding state
%vector $\S$.  
% So that the power and average distortion
% constraints are satisfied,
A randomly generated codeword 
$\hat{\s}$ is admitted only if it satisfies the empirical 
probabilities $\mathcal{P}_{\hat{\s}|\s}(0|1) = \eo$ and 
$\mathcal{P}_{\hat{\s}|\s}(1|0) = \ei$, where $\eo$ and $\ei$ are
fixed and  $q(1-\eo) + (1-q)\ei = p$.
% The random candidates for power loading vectors 
% $\hat{\s}$ are generated until an {\em admissible} codeword is found.
 This will ensure that 
averaged over all the state vectors, 
the number of sub-channels used for transmission are given 
by $pN$ and average number of unused
good sub-channels are $q\eo N$. Next, we shall find the
expected variable-length encoder rate, averaged over 
this ensemble of codes, and then, by usual argument, we can assert that 
there must exist at least one set of $\{\hat{\s}\}$ that gives the 
performance as good as the average.

Let $L$ represent the random variable 
denoting the fraction of ones in the state vector $\S$.
Clearly, $H(L) \le \log_2(N+1)$.
If a variable-length feedback codebook with average rate of 
$R_f$ bits per sub-channel per coherence block is used, then
we can write
$N R_f \le H(\hat{\S}) + 1$.
The rate can be further upper bounded as
$N R_f \le H(\hat{\S}, L) + 1 = H(\hat{\S}|L) + H(L) + 1$.
%The conditional entropy term is further computed as,
%\begin{eqnarray}
%H(\hat{\S}|B) &=& \sum_{b =0}^{N} \Pr\{b\} H(\hat{\S}|b)\\
%&=& \sum_{\s} \Pr\{\s\} H(\hat{\S}|\s)\\
%&=& H(\hat{\S}|\S)
%\end{eqnarray}
Averaging over the random code book selection
we get that,  
\begin{equation}\label{eq:varrate1}
N R_f \le E_{\hat{\S}}[ H(\hat{\S}|L)] + H(L) + 1.
\end{equation}
Corresponding to a channel state vector with 
$L = l$, define
$q_{l}$ as the probability  that a randomly 
drawn codeword $\hat{\s}$ is admissible. 
%If $\mathcal{P}_{\s}(1) = b$
Then we have,
\begin{equation}\label{eq:qs}
q_{l} = {lN \choose \eo lN} {(1-l)N \choose \ei (1-l)N}\, p^{n(l)N} (1-p)^{(1-n(l))N},
\end{equation} 
where $n(l) = (1-\eo) l + \ei (1-l)$. Further, it is argued in \cite{Pinkst67}
that given the geometric distribution $p(k|l) = q_{l} (1- q_{l})^{k-1}$
we have
\begin{align}
E_{\hat{\S}}[ H(\hat{\S}|l)] &= - \sum_{k =1}^{\infty} p(k|l) \log_2 p(k|l)\\
\label{eq:varrate2}
&\le - \log_2 q_{l} +\log_2e.
\end{align}
Combining \eqref{eq:varrate1} and  \eqref{eq:varrate2} and 
using the fact that $H(L) \le \log_2(N+1)$ we have,
 \begin{equation}\label{eq:varrate3}
N R_f \le -\sum_{l=0}^{N} (\log_2 q_l) \Pr\{L =l\} + \log_2e+ \log_2(N+1) + 1.
\end{equation}
Further, applying Robbin's approximation \eqref{eq:rob} to \eqref{eq:qs} we have
\begin{equation}\label{eq:qsub}
  \begin{split}
-\log_2(q_l) \le -lNH_2(\eo) &- (1-l)NH_2(\ei) \\
+& n(l)N\log_2(p) + (1-n(l))N\log_2(1-p) + \log_2 N + K_3,
  \end{split}
\end{equation}
where $K_3 = \frac{1}{2}\log_2[e^8(2\pi)^2\eo\ei(1-\eo)(1-\ei)]$.
Substituting \eqref{eq:qsub} into \eqref{eq:varrate3} and using the fact that 
$E[n(l)] = p$ and $E[l] = q$, we get, 
\begin{equation}\label{eq:RfNvar}
 R_f \le  i(\eo, \ei) + \frac1N (\log_2e+K_3+\log_2 N +\log_2(N+1) +1) \ .
\end{equation}
Since this rate exceeds $i(\eo,\ei)$, similar to 
Appendix \ref{ap:prop:bin-fix}, we can introduce additional
distortion so that $\eo$ and $\ei$ are replaced by 
$\eo + \delta$ and $\ei + \frac{q \delta}{1-q}$. Therefore,
using \eqref{eq:temp-ineq2}, for large enough $N$ and 
small enough $\delta$, \eqref{eq:RfNvar} yields
\begin{equation}
R_f \le  i(\eo, \ei) - \frac12 q[H_2'(\eo) + H_2'(\ei)] \delta + \frac{3\log_2 N}{N}
\end{equation}
Finally, choosing $\delta = \frac{6 \log_2 N}{q[H_2'(\eo) + H_2'(\ei)] N }$ gives 
$R_f \le i(\eo, \ei)$ and capacity as \eqref{eq:cap-var}.

\section{Derivation of $i_l(q_{00},q_{01},q_{10},q_{11})$ and $i_u(\eo, \ei)$}
\label{ap:corr-rd}

The lower bound can be explicitly computed as follows
\begin{align}
\label{eq:rrlb}
i_l(q_{00},q_{01},q_{10},q_{11}) &= I(S_1;\hat{S}_1|S_0)\\
&= H(\hat{S}_1|S_0) - H(\hat{S}_1|S_1, S_0).
\end{align}
Each entropy term is further computed as
\begin{equation}\label{eq:ent1}
  H(\hat{S}_1|S_0) = (1-q) H_2\left(\frac{q_{00}+q_{01}}{1-q}\right) + q H_2\left(\frac{q_{10}+q_{11}}{q}\right)
\end{equation}
and
\begin{equation}\label{eq:ent2}
  \begin{split}
    H(\hat{S}_1|S_1, S_0) = (1-q) & (1-\delta_{01})
    H_2\left(\frac{q_{00}}{(1-q)(1-\delta_{01})}\right) + q
    \delta_{10} H_2\left(\frac{q_{10}}{q \delta_{10} }\right) \\
    &+ (1-q)\delta_{01}
    H_2\left(\frac{q_{01}}{(1-q)\delta_{01}}\right) + q(1-\delta_{10})
    H_2\left(\frac{q_{11}}{q(1-\delta_{10})}\right).
  \end{split}
\end{equation}
%which gives,
%\begin{multline}
%\overline{R}_{c}^{lb} = (1-q) H_2[\delta_{01}(1-\epsilon_0) + (1-\delta_{01})\epsilon_1] \\
%+ qH_2[(1-\delta_{10})(1-\epsilon_0) + \delta_{10}\epsilon_1] - qH_2[\epsilon_0] - (1-q)H_2[\epsilon_1].
%\end{multline}
Next we compute the upper bound $f_2(\eo, \ei)$. Recall that, 
in order to arrive at the upper bound, we have 
assumed that conditioned on $S_i$, $\hat{S}_i$ is 
independent of all other elements in $\S$, thus
\begin{align}
\label{eq:rrub}
i_u(\eo, \ei) &= I(S_1;S_2,\hat{S}_1|S_0)\\
&= H(S_2,\hat{S}_1|S_0) - H(S_2, \hat{S}_1|S_1, S_0)\\
&= H(S_2|S_0) + H(\hat{S}_1|S_0,S_2) - H(\hat{S}_1|S_1) - H(S_2|S_1).
\end{align}
Each entropy term can be further computed as
\begin{align}
\label{eq:ent3}
H(S_2|S_0) & = qH_2\left((1-\dio)^2 + \dio \doi\right) + (1-q) H_2\left(\doi(1-\dio) + (1-\doi)\doi\right),\\
\label{eq:ent4}
H(\hat{S}_1|S_1) & = qH_2(\eo) + (1-q) H_2(\ei), \\
\label{eq:ent5}
H(S_2|S_1) & = qH_2(\dio) + (1-q)H_2(\doi)
\end{align}
and
\begin{equation}\label{eq:ent6}
  \begin{split}
    H(\hat{S}_1|S_0,S_2) = &\left((1-\doi)^2(1-q)+ \dio^2q\right)
    H_2\left(w_{00}\right)\\
    & \qquad +2\left(\doi(1-\doi)(1-q)
      +\dio(1-\dio)q\right)H_2\left(w_{01}\right)\\
    & \qquad\qquad+\left(\doi^2(1-q)+(1-\dio)^2q\right)H_2\left(w_{11}\right),
  \end{split}
\end{equation}
where the probabilities in the argument of binary entropy
functions are defined as,
\begin{equation} \label{eq:w01} 
w_{s_0s_2} = P_{\hat{S}_1|S_0,S_2}(0|s_0,s_2), \quad s_0,s_2=0 \text{ or } 1.
\end{equation}
We note that $w_{10} = w_{01}$ and the probabilities can be explicitly
computed as
\begin{align}
w_{00} & = \frac{(1-\ei)(1-\doi)^2(1-q)+ \eo \dio^2q}{(1-\doi)^2(1-q)+ \dio^2q},\\
w_{01} & = \frac{(1-\ei)\doi(1-\doi)(1-q)+\eo \dio(1-\dio)q}{\doi(1-\doi)(1-q)+\dio(1-\dio)q},\\
w_{11} & = \frac{(1-\ei)\doi^2(1-q)+\eo(1-\dio)^2q}{\doi^2(1-q)+(1-\dio)^2q} \ .
\end{align}

\section{Proof Sketch of Proposition \ref{Cap-CSI}}
\label{ap:prop-grp}

Consider the maximization of the capacity $\Cg(m, t)$ in
\eqref{eq:mc-cap} over the group size $m$ and threshold $t$ subject to
$G H(q) \le B$.  The capacity can be bounded as
\begin{equation}  \label{eq:lCu}
  Nq \log \left(1+ \frac{P t}{N q}\right)
  \le C(m,t)
  \le Nq \log \left(1+ \frac{P (t+1)}{N q }\right).
\end{equation}
The lower bound is simply by observing that the logarithm term in
\eqref{eq:mc-cap} takes on its minimum value at the boundary $\tau=t$,
whereas the exponential term integrates to 1.  The upper bound can be
shown using the fact that $\int_0^\infty e^{-x} \log(x + a) \,\intd
\tau < \log(1+a)$ for all $a>0$.
\begin{comment}
\begin{equation}
  \int_0^\infty e^{-\tau} \log \frac{\tau + a}{1+a} \,\intd \tau < 0.
\end{equation}
 $\Cg^{lb} \le \Cg(m, t)
\le \Cg^{ub}$, where
\begin{align}\label{eq:Cub-lb}
\Cg^{lb} &= Nq \log \left(1+ \frac{P t}{N q}\right) \,\,\, \textrm{and} \\
\Cg^{ub} &= Nq \log \left(1+ \frac{P (t+1)}{N q }\right).
\end{align}
\end{comment}
Clearly, the maximum value of $C(m,t)$ subject to $G H(q) \le B$ is no
greater than the maximum value of the upper bound of $C(m,t)$ in
\eqref{eq:lCu} subject to the same constraint.
% the solution to the optimization problem: $\max \,\Cg^{ub},
% \,\textrm{subject to}\,\,\, -Gq \log q \le B$.
Next we obtain a solution to the latter optimization problem and show that 
for large $B$ and $N$, it provides a good approximation to the 
solution of the original optimization problem.  
Without loss of generality, substituting $w = N q$ the optimization problem 
can be written as
\begin{equation}\label{eq:ubopt}
\max_{w, t}\,\, 
%\Cg^{ub}
\bar{C} = w \log \left(1+ \frac{P (t+ 1)}{w}\right), \,\, \textrm{subject to:} \,\, w\, t \le B.
\end{equation}

Assuming that the feedback constraint 
is tight, i.e., $w\, t = B$, the 
optimum $w$ must satisfy
\begin{equation}\label{eq:dervw}
(1+u + P/w) \log (1+u + P/w) = (1+2 u + P/w),
\end{equation}
where $u = P B/{w^2}$.
A closed-form solution to \eqref{eq:dervw} seems difficult, 
however, insight can be obtained by assuming that $N, B$ are large.
In addition, we assume that the optimal $w$ increases with $B$ 
such that $u \gg P/w$ or, equivalently, as $B \to \infty$, 
$w/B \to 0$. We will see later that 
this is indeed true. Therefore, observing that the $P/w$ terms in \eqref{eq:dervw}
are small compared to $u$, the optimal $w^\star = \sqrt{\frac{PB}{u^\star}} + o(1)$,
where $o(1)$ is vanishingly small as $B \to \infty$, and 
$u^\star$ is the solution to $(1+u^\star) \log (1+u^\star) = (1+2 u^\star)$.
Since we assume that $w t = B $, 
solving \eqref{eq:ubopt} gives
\begin{align}
\label{eq:thres1}
t^\star &= \sqrt{\frac{u^\star B}{P}} + o(1)\\
\label{eq:NG1}
m^\star &= \sqrt{\frac{P}{u^\star B}} \log \frac{N}{\sqrt{P\,B/u^\star}} + o(1)\\
\label{eq:cub1}
%\Cg^{ub} 
\bar{C} &= \sqrt{\frac{P B}{u^\star}} \log (1 + u^\star)  + o(1) .
\end{align}
The parameter values
satisfy the original feedback constraint $G H(q) = B$  
and in fact the lower bound in \eqref{eq:lCu} also behaves 
as  \eqref{eq:cub1} (although the value associated with the 
$o(1)$ term change). This implies that 
the optimal parameters that maximize the capacity $C(m,t)$ 
satisfy \eqref{eq:thres1} and \eqref{eq:NG1}, and the 
capacity $\Cg^\star$ is approximated by \eqref{eq:cub1} to within a
vanishingly small term.

Note that \eqref{eq:NG1} implies that for $m^\star > 1$ we should 
have feedback in the range $B < \frac{P}{u^\star} (\log N)^{2-\eta_1}$
for large $N$ where  $\eta_1 \in (0,2)$ is the solution to 
\begin{equation}\label{eq:eps1-def}
 \log N - \log \left[\frac{P}{u^\star} (\log N)^{1-\frac{\eta_1}{2}}\right] = (\log N)^{1-\frac{\eta_1}{2}}.
\end{equation}
It is easy to see that $\eta_1 \to 0$ as $N \to \infty$. Next we solve for the 
optimal parameters when $B \ge \frac{P}{u^\star} (\log N)^{2-\eta_1}$. 
Again we first solve the upper bound 
maximization problem \eqref{eq:ubopt} with 
$m = 1$ or, equivalently $q = e^{-t}$.
Namely
\begin{equation}\label{eq:ubopt2}
  \max_{t}\, 
  % \Cg^{ub}
  \bar{C} = N e^{-t} \log \left(1+ \frac{P (t+1)}{N e^{-t}}\right),
  \,\, \textrm{subject to:} \,\,  N t e^{-t}\le B.
\end{equation}
Assuming that the feedback constraint is tight, i.e.,
 $N t e^{-t} = B$, 
we get the optimal threshold and upper bound on capacity 
\begin{align}
\label{eq:thres2}
t^\star &= \log \frac{N \log N}{B} + o(1)\\
\label{eq:cub2}
%\Cg^{ub}
\bar{C} &= \frac{B}{\log N} \log\left(1 + \frac{P \log N}{B}  
  \log\frac{N \log N}{B}\right) + O(1)
\end{align}

Again, it can be checked that with appropriate adjustments 
to the $o(1)$ and $O(1)$ terms in \eqref{eq:thres2} and \eqref{eq:cub2}, 
respectively, the threshold \eqref{eq:thres2}
satisfies the original feedback constraint $N H(e^{-t}) = B$ and the 
lower bound in \eqref{eq:lCu}
also behaves as  \eqref{eq:cub2}. This implies that 
the threshold, which maximizes the capacity $\Cg$ in the feedback range 
$ B \ge \frac{P}{u^\star} (\log N)^{2-\eta_1}$ satisfies
 \eqref{eq:thres2}, $m^\star = 1$ and the 
capacity $\Cg^\star$ is also given by \eqref{eq:cub2}.

Furthermore, note that as $B$ increases, the threshold 
\eqref{eq:thres2} decreases. However, decreasing the threshold 
beyond a certain optimal value decreases the capacity upper 
bound in \eqref{eq:ubopt2}. 
The optimum value of the threshold that maximizes  
the upper bound in \eqref{eq:ubopt2} is given by \eqref{eq:optto3} and 
the corresponding upper bound is given in \eqref{eq:Cg}, corresponding to 
$B > B_{max}$,
%\begin{align}
%\label{eq:optto2}
%& t^\star = \left[ \log N - (1+\eta_2)\log\log N -\log P \right] + o(1)\\
%\label{eq:rateWcsi2}
%& \Cg^{ub}  = P \left[\log N - (1+\eta_2)\log\log N \right] + O(1) .
%\end{align}
where $\eta_2 \in (0,1)$ is the solution to 
\begin{equation}\label{eq:eps2-def}
 \log N - \log [P (\log N)^{1+\eta_2}] = (\log N)^{(1+\eta_2)/2}. 
\end{equation}
Clearly, $\eta_2 \to 1$ as $N \to \infty$. Substituting $m = 1$
and \eqref{eq:optto3} into
% $\Cg^{lb}$ in \eqref{eq:Cub-lb} 
the lower bound in \eqref{eq:lCu} gives 
that the lower bound and 
hence the capacity also behave as in \eqref{eq:Cg}.
% (although the value associated with the $O(1)$ term change).
 Therefore, the optimal threshold that maximizes the
capacity is given by \eqref{eq:optto3} with adjusted $o(1)$ term. 
Corresponding to \eqref{eq:optto3}, the maximum required feedback is given by 
$B_{max} = N H(e^{-t}) = P(\log N)^{2+\eta_2} + o(\log^2N)$.

\bibliographystyle{IEEEtran}
%\bibliography{def,csf-journ1,alljab}
\bibliography{def,dguo,alljab,csf-journ1}
\end{document}